\documentclass{amsart}

\usepackage{amsmath,amssymb,amsfonts, wasysym}
\usepackage[utf8]{inputenc}

\usepackage{amsthm, amsmath}
\usepackage{url}

\usepackage{multicol}

\usepackage{graphicx}
\usepackage{caption}
\usepackage{subcaption}

\usepackage{thmtools,thm-restate}

\newtheorem{lemma}{Lemma}
\newtheorem{definition}{Definition}
\newtheorem{theorem}{Theorem}

\allowdisplaybreaks

\begin{document}
\title{Correctness of the Chord Protocol}

\author{Bojan Marinković}
\address{Mathematical Institute of the Serbian Academy of Sciences and Arts, Belgrade, Serbia}
\email{\url{bojanm@mi.sanu.ac.rs}}

\author{Zoran Ognjanović}
\address{Mathematical Institute of the Serbian Academy of Sciences and Arts, Belgrade, Serbia}
\email{\url{zorano@mi.sanu.ac.rs}}

\author{Paola Glavan}
\address{Faculty of Mechanical Engineering and Naval Architecture, University of Zagreb, Zagreb, Croatia}
\email{\url{pglavan@fsb.hr}}

\author{Anton Kos}
\address{Faculty of Electrical Engineering, University of Ljubljana, Ljubljana, Slovenia}
\email{\url{anton.kos@fe.uni-lj.si}}

\author{Anton Umek}
\address{Faculty of Electrical Engineering, University of Ljubljana, Ljubljana, Slovenia}
\email{\url{anton.umek@fe.uni-lj.si}}

\maketitle
\begin{abstract}
Internet of Things (IoT) can be seen as a cooperation of the various heterogeneous devices with limited performances, that participate in the same system. By they nature, these devices can be very distributed. The core of every IoT system is its discovery and control service.
The Chord protocol is one of the first, simplest and most popular distributed protocol and can be use as a backbone of the discovery and control services of an IoT system.
In this paper we prove the correctness of the Chord protocol using the
logic of time and knowledge. We consider Chord actions that
maintain ring topology with the additional assumption the nodes are not
allowed to fail or leave.

\end{abstract}

\section{Introduction}

Internet of Things (IoT) paradigm can be defined as \cite{iotdef}: "The pervasive presence around us of a variety of things or objects which, through unique addressing schemes, are able to interact with each other and cooperate with their neighbors to reach common goals." In this framework the smart
objects, which are connected by an Internet-like structure, are able to communicate and exchange information and to enable new forms of interaction among things and people \cite{CDFLMPV14}. The core of every IoT system consists of its discovery and control service.  Usually, the objects which participate in an IoT system have limited computing power, memory and power supply. It is the common thing that various heterogeneous devices participate in the same IoT system. Ordinarily, these devices are highly distributed, so they participate in a distributed, i.e. Peer-to-Peer (P2P), system.

In a homogeneous decentralized P2P system \cite{P2P2010}, many nodes (peers) execute the same application, and have equal rights during that execution. They might join or leave system at any time. In such a framework processes are dynamically distributed to peers, with no centralized control. Thus,  P2P systems have no inherent bottlenecks and can potentially scale very well. Also, those systems are resilient to failures, attacks, etc., since there are no nodes which perform the critical functions of the systems. The main applications of P2P-systems involve: file sharing, redundant storage, real-time media streaming, etc.

P2P systems are frequently implemented in a form of overlay networks \cite{p2p-book}, a structure that is totally independent of the underlying network that is actually connecting devices. Overlay networks represent a logical look on organization of the system resources.
Some of the overlay networks are realized in the form of a Distributed Hash Tables (DHTs), which provides a lookup service similar to a hash table; $\langle key, value\rangle$ pairs are stored in a DHT, and any participating node can efficiently retrieve the value associated with a given key.
Note that \textit{key} is not used as a cryptographic notion, but (following the common practice in DHT-related papers) to represent identifiers of objects. Responsibility for maintaining the mapping
from keys to values is distributed among the peers, in such a way that any change in the set of participants causes a minimal amount of disruption. 
The Chord protocol \cite{Chord,Chord-TR,Chord-IEEE} is one of the first, simplest and most popular DHTs.
The paper \cite{Chord} which introduces Chord was recently awarded the SIGCOMM 2011 Test-of-Time Award.

Because of the simplicity and popularity of the Chord protocol, it was used for the realization of the discovery and/or control service of IoT systems described in \cite{BUC13, CDFLMPV14, EFKS10, PP12, XWWZQZ15}.

As we mentioned above, the discovery and control services are cores of an IoT system, and because of that, in this paper we will prove the correctness of the Chord protocol using the
logic of time and knowledge. We consider the case when the nodes are not
allowed to fail or leave and concern Chord actions that
maintain ring topology. 

We are aware of only a few attempts to formally verify behavior of
DHTs and particularly Chord \cite{BG06, BG07, krishna, liben02, alloy}. We consider them below and compare with our
approach.

The rest of the paper is organized in the following way: in Section \ref{related} we consider other approaches for proving the correctness of the Chord protocol and clearly present the contributions of this paper; Section \ref{chord} presents a short description of the Chord protocol; in Section \ref{logic} we present a logical framework which will be used to prove the correctness of the maintenance of the ring topology of the Chord protocol  with the respect of the fact that nodes are not allowed to departure the system after they join it; the proof is given in Section \ref{proof}; we conclude with Section \ref{conclusion}. In \ref{appendix} we provide detailed proofs of most lemmas and theorems from the paper. 

\section{Related Work and Contributions}\label{related}

\subsection{Related Work}

The Chord protocol is introduced in
\cite{Chord,Chord-TR,Chord-IEEE}. The papers analyze the protocol,
its performance and robustness under the assumption that the
nodes and keys are randomly chosen, and give several theorems that
involve the phrase {\em with high probability}, for example: "With
high probability, the number of nodes that must be contacted to
find a successor in a $N$-node network is $O(\log N)$".

The only statement in the papers \cite{Chord,Chord-TR,Chord-IEEE} which avoids the mentioned
phrase about high probability is Theorem IV.3. It corresponds to
our Lemma \ref{joinnn} and proves that inconsistent states
produced by executing several concurrent joins of the new nodes are transient,
i.e., that after the last node joins the network will form
a cycle. 
More general sequences of concurrent joining and leaving are
considered in \cite{liben02}, where a lower bound of the rate at
which nodes need to maintain the system such that it works
correctly is given with high probability. In this paper we are not considering possible failures and leaves of the nodes. Our intention is that include this segment in our future work.

Anyway, it is not quite clear how to compare these two approaches (deterministic and probabilistic), but in
our opinion there is benefit from both of them. One can argue that
the probabilistic approach, i.e. providing lower bounds of
probabilities, is useful to study robustness of protocols. On the
other hand, it will be useful to describe sequences of actions leading
to (un)stable states of Chord networks, to be able to 
analyze properties of systems that incorporate Chord and
assume its correctness, as it is the case with the discovery and/or control service of an IoT system.

In \cite{krishna} the theory of stochastic processes is used to
estimate the probability that a Chord network is in a particular
state.
In \cite{BG06,BG07} Chord's stabilization algorithm is modelled
using the $\pi$-calculus and it's correctness is established by
proving the equivalence of the corresponding specification and
implementation. Possible departures of nodes from a network are
not examined in this approach.
In \cite{alloy} the Alloy formal language is used to prove
correctness of the pure join model. The same formalization
present several counterexamples to correctness of Chord
ring-maintenance in the general case. 

In \cite{halpern} a joint frame for reasoning about knowledge and linear time is presented, and the  proof of weak completeness for a logic which combines expressions about knowledge with linear time is provided.

As we mentioned in Introduction using DHT or Chord in IoT domain is not a novelty \cite{BUC13, CDFLMPV14, EFKS10, PP12, XWWZQZ15}.
In \cite{BUC13} authors proposed distributed control plane. They consider the problem how to deliver control messages to the devices that are in sleeping mode most of the time. Proposed DHT algorithm is Chord.
The paper \cite{CDFLMPV14} introduce scalable, self-configuring and automated service and resource discovery mechanism based on structured DHT architecture.
The article \cite{EFKS10} presents comparison of the discovery service mechanisms in IoT domain, both traditional and distributed approaches.
In \cite{PP12} authors give the description of a novel discovery service for IoT which adopts DHT approach with multidimensional search domain.
Authors of \cite{XWWZQZ15} presented discovery service for objects carrying RFID tags based on double Chord ring. 
In all these articles, the correctness of the Chord protocol was accepted for granted.

\subsection{Contributions}
In this paper we:
\begin{itemize}
\item provide axiomatization and prove the soundness, strong completeness and decidability of the logic of time and knowledge;
\item describe the Chord protocol using the logic of time and knowledge;
\item prove the correctness of the maintenance of the ring topology of the Chord protocol with the respect of the fact that nodes are not allowed to departure the system after they join it.
\end{itemize}

This work is motivated by the importance of the discovery and control service of an IoT system and the obvious fact that errors in concurrent
systems are difficult to reproduce and find merely by program
testing. This proof could be, also, the foundation for the formal proof 
created using a formal proof assistant (like, Coq or Isabelle/HOL).

\section{Chord Protocol}\label{chord}

The papers \cite{Chord,Chord-TR,Chord-IEEE} introduce the Chord protocol and give the
specification of it in C$++$-like pseudo-code. They present the correctness, performance and robustness of the Chord protocol. Here, we will provide a short description of it.

A number of nodes running the Chord protocol form a ring-shaped network. The main operation supported by Chord is mapping the given key onto a node using consistent hashing.The consistent hashing \cite{conshash} provides load-balancing, i.e., every node receives roughly the same number of keys, and only a few keys are required to be moved when nodes join and leave the network. Chord networks are overlay systems. Thus, each node in a network, that consists of $N$-nodes, needs ``routing'' information about only a few other nodes, $O(\log N)$, and resolves all lookups via $O(\log N)$ messages to other nodes.

\begin{figure}[h]
\centering
\includegraphics[width=0.55\textwidth]{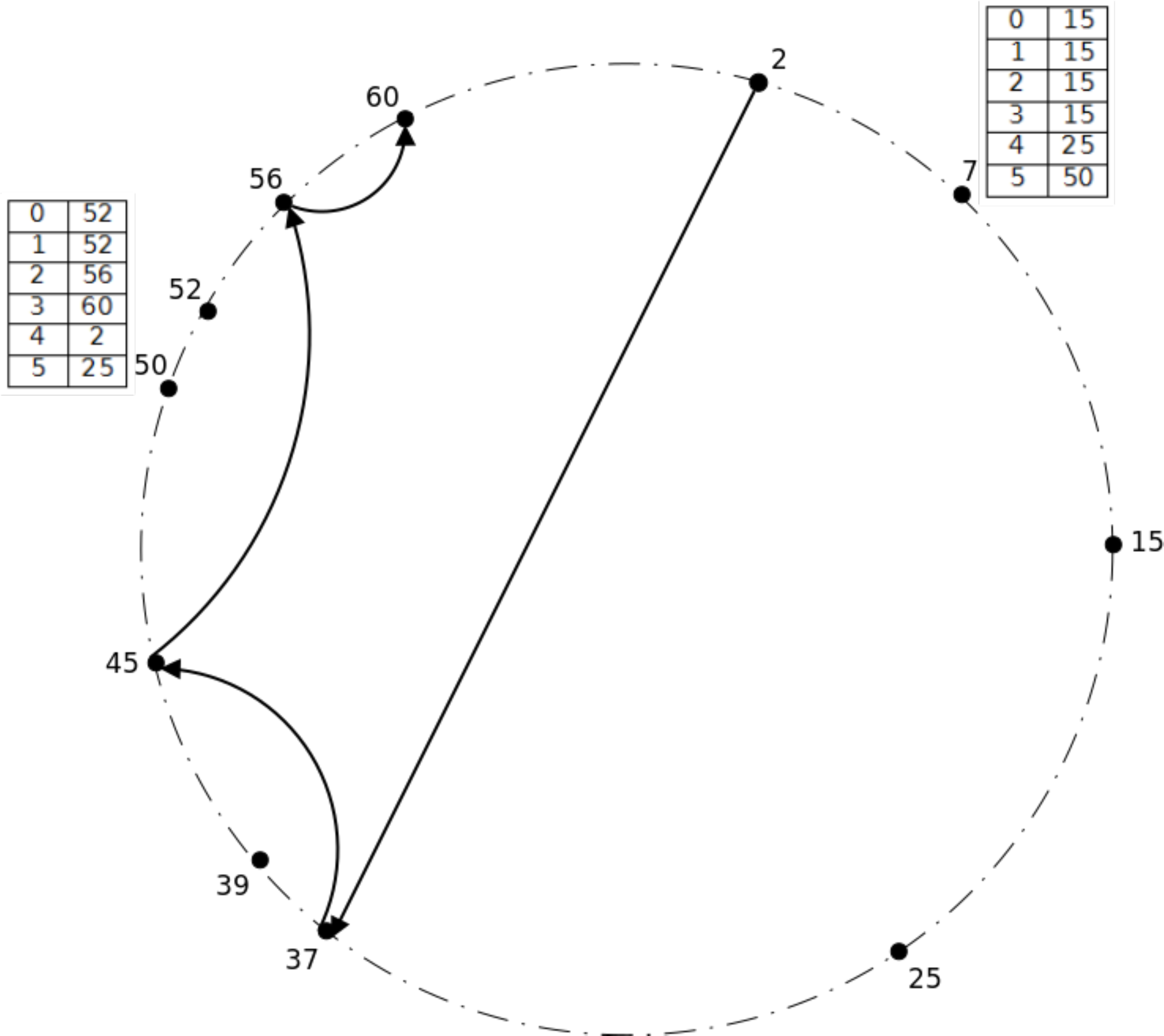}

\caption{Chord lookup procedure}\label{chord_lookup}
\end{figure}

As it is shown, the Chord's stabilization algorithm maintains good
lookup performance despite continuous failure and joining of nodes.
When the network is not stable, i.e., the corresponding ``routing'' information is out of date since nodes join and leave arbitrarily,
the performance degrades.

Identifiers are assigned to nodes and keys by the consistent hash function. The identifier for a node or a key, $hash(node)$ or
$hash(key)$, is produced by hashing IP of the node, or the value of the key.
The length of identifiers, for example $m$ bits), must guarantee that the probability that two objects of the same type are assigned
same identifiers is negligible.
Identifiers are ordered in an identifier circle modulo $2^m$. Then, the key $k$ is assigned to the node such that
$hash(node)=hash(key)$. If such a node does not exist, the key is assigned to the first node in the circle whose identifier is greater
than $hash(key)$.

Every node possesses information on its current successor and predecessor nodes in the identifier circle. To accelerate the lookup
procedure, a node also maintains routing information in the form of the so-called {\em Finger Table} with up to $m$ entries. The $i^{th}$
entry in the table at the node $n$ contains the identifier of the first node $s$ that succeeds $n$ by at least $2^{i-1}$ in the
identifier circle, i.e., $s= successor(n+2^{i-1})$, where $1 \leqslant i \leqslant m$, and all arithmetic is preformed modulo $2^{m}$. Figure \ref{chord_lookup} presents Finger tables of nodes $n_7$ and $n_{50}$.

One node can be aware of only a few other nodes in the system, like node $n_7$ from Figure \ref{chord_lookup} knows for the existence of only 3 other nodes. Some other can have different node identifier in almost every entry in its Finger table, like node $n_{50}$ from Figure \ref{chord_lookup}.

During the lookup procedure, a node forwards a query to the largest element of the Finger table which is smaller than the key used in the query, respect to the used arithmetics. In the example illustrated with Figure \ref{chord_lookup}, if $n_2$ is looking for the responsible node for the key with identifier $57$, it will forward this query to node $n_{37}$, the closes node from its finger table. After, that this query will be forwarded to  $n_{45}$ and $n_{56}$, until it finally ends at $n_{60}$. The answer if $n_{60}$ contains the key and respected value with identifier $57$ will be returned to node that started query, in this case $n_2$.

The
stabilization procedure implemented by Chord must guarantee that each node's finger table, predecessor and successor pointers are up to date. The
procedure runs periodically in the background at each node.
To increase robustness, each Chord node can create a successor list of size $r$, containing the node's first $r$ successors.

\begin{figure}[h]
\begin{center}
\includegraphics[width=0.18\textwidth]{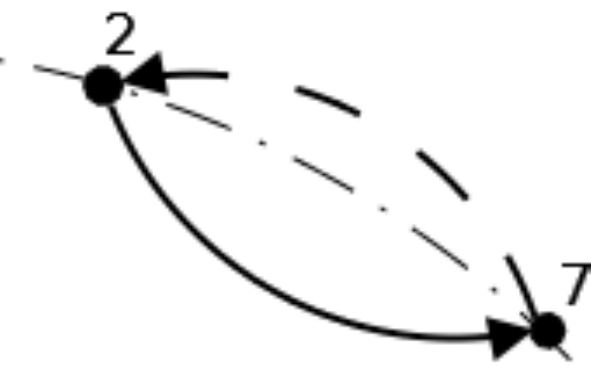}
\includegraphics[width=0.18\textwidth]{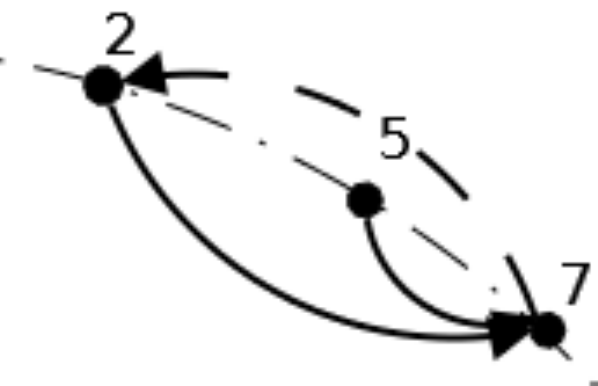}
\includegraphics[width=0.18\textwidth]{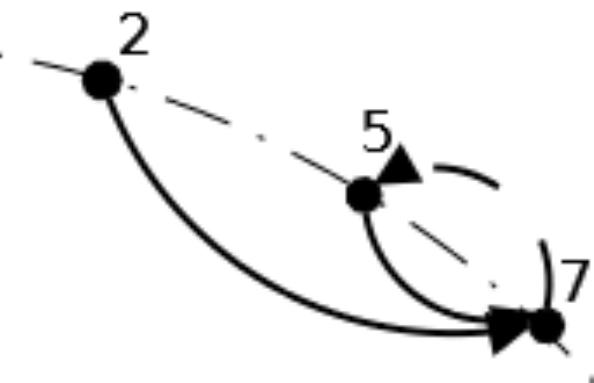}
\includegraphics[width=0.18\textwidth]{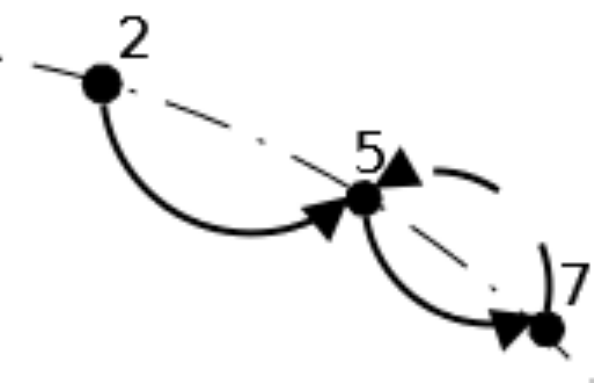}
\includegraphics[width=0.18\textwidth]{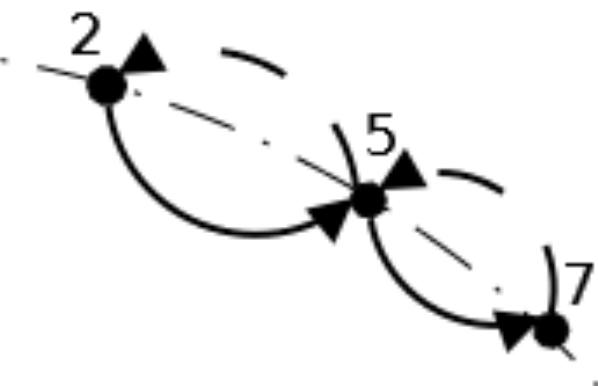}

\caption{Stabilization during the joining of a new node}\label{new_node_join}
\end{center}
\end{figure}

Figure \ref{new_node_join} illustrates the process of joining of the node $n_5$ between nodes $n_2$ and $n_7$. As a first step $n_5$ will set its successor to $n_7$. During the stabilization process $n_7$ will set its predecessor to $n_5$, then $n_2$ will set its successor to $n_5$ and, finally, $n_5$ will set its predecessor to $n_2$.  

Beside the mapping of keys onto the set of nodes, the only other operations
realized by Chord are adding a node to network or removing  a node from a network.
When a node $n$ joins an existing network, certain keys previously
assigned to $n$'s successor now become assigned to $n$. When a
node $n$ leaves the network regularly, it notifies its predecessor
and successor and reassigns all of its keys to the successor.

\section{Logic of Time and Knowledge}\label{logic}

As we mentioned in the previous Section, a system which runs the Chord protocol is a dynamic multi-agent system, where every agent has it own partial view of the surrounding environment. To be able to reason about such system, we need to introduce a framework for formal description of changes of the knowledge of an agent during the time. In this section we present logic of time and knowledge.

\subsection{Syntax}
\noindent

Let $\mathbb{N}$ be the set of non-negative integers. We denote $\mathbf{N}=\{n_0,\ldots n_{m-1} \}$, where $m \in \mathbb{N}$, and then let $\mathbf{N_1} = \mathbf{N} \cup \{u\}$ be the set of propositional variables.

The set $For$ of all formulas is the smallest superset of $\mathbf{N_1}$ which is closed under the following formation rules:

\begin{itemize}
	\item $\langle \phi, \psi \rangle \mapsto \phi \ast \psi$ where $\ast \in \{\succ, \prec\}$ and $\phi, \psi \in \mathbf{N_1}$,
	\item $\langle \phi, \psi, \varphi \rangle \mapsto \phi \mathtt{M} \langle \psi, \varphi \rangle $ where $\phi, \psi,  \varphi \in \mathbf{N}$,
	\item $\psi \mapsto \ast \psi$ where $\ast \in \{\neg, \bigcirc, \CIRCLE, \mathtt{K}_i\}$,
	\item $\langle \phi, \psi \rangle \mapsto \phi \ast \psi$ where $\ast \in \{\wedge, \mathtt{U}, \mathtt{S}\}$. 
\end{itemize}

The operators $\succ$ and $\prec$ represent relations successor and predecessor of a node. The tip of the "arrow" is pointing to the node with "greater" identifier, with respect to the ordering determined by the ring shaped Chord network.
We will use abbreviation $n_i \succ^2 n_k$ for $n_i, n_k \in N$ iff there is an $n_j \in N$ such that $n_i \succ n_j$ and $n_j \succ n_k$, and $n_k \prec^2 n_i$ for $n_i, n_k \in N$ iff there is an $n_j \in N$ such that $n_k \prec n_j$ and $n_j \prec n_i$. Similarly, we can define $n_j \succ^i n_k$, as well as $n_j \prec^i n_k$ for $n_j, n_k \in N$ and $0< i < m$. Figure \ref{succprec} illustrates the relations $\succ$, $\prec$ (Figure \ref{chordsuccpred}) and $\succ^i$ (Figure \ref{chordsucci}). 

\begin{figure}[h]
\centering
\begin{subfigure}{0.45\textwidth}
\centering
\includegraphics[height=3cm]{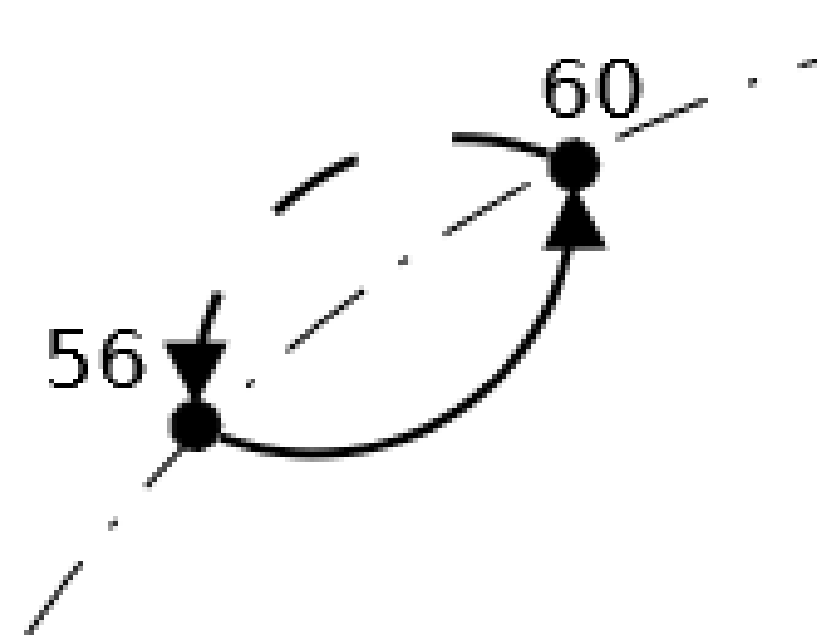}
\caption{$n_{56} \succ n_{60}$ and $n_{60} \prec n_{56}$}\label{chordsuccpred}
\end{subfigure}
\begin{subfigure}{0.45\textwidth}
\centering
\includegraphics[height=3cm]{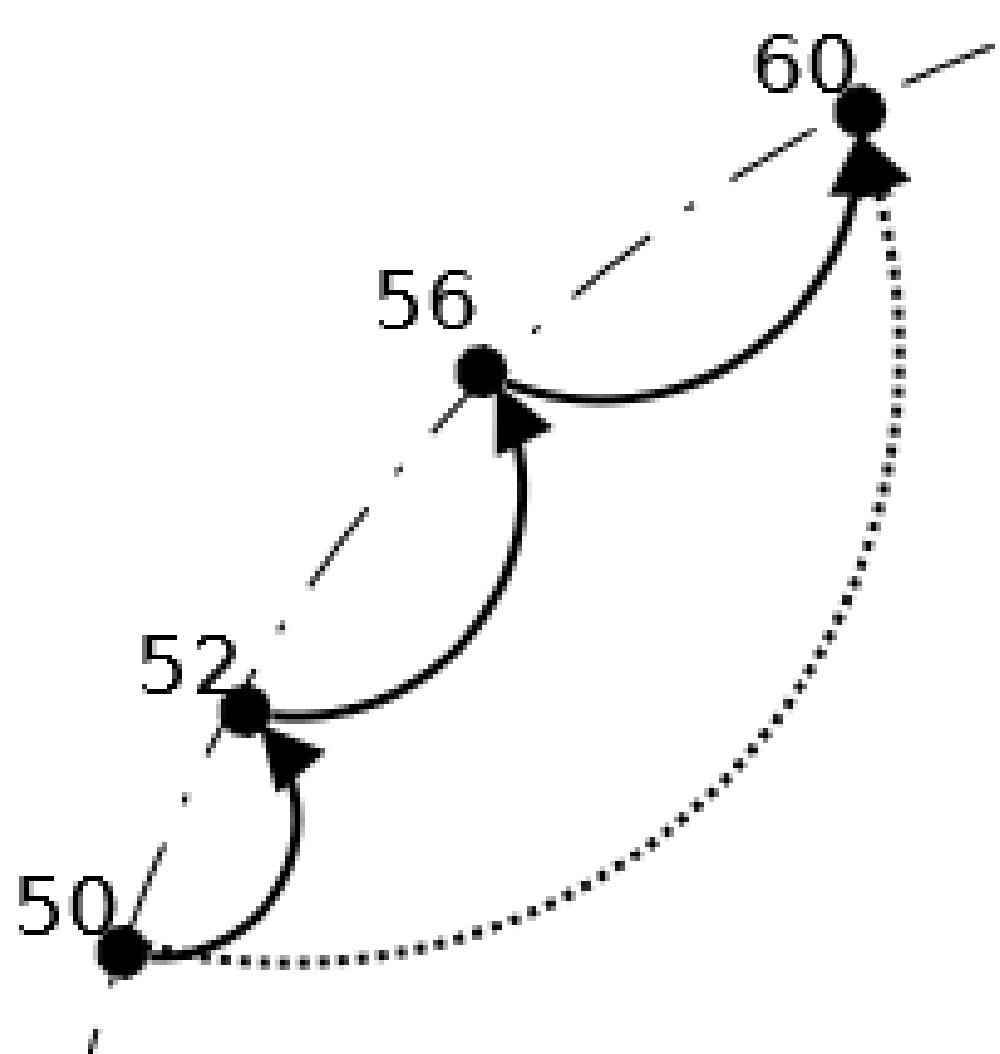}
\caption{$n_{50} \succ^3 n_{60}$}\label{chordsucci}
\end{subfigure}
\caption{Examples of $\succ$, $\prec$ and $\succ^i$}\label{succprec}
\end{figure}

The operators $\neg$ and $\wedge$ represent standard logical negation and conjunction. The operators $\bigcirc$, $\CIRCLE$, $\mathtt{U}$ and $\mathtt{S}$ are standard temporal operators Next, Previous, Until and Since. The operator $\mathtt{K}_i$ represents the knowledge of the agent $i$.

The remaining logical and temporal connectivities $\vee$, $\rightarrow$, $\leftrightarrow$, $\mathtt{F}$, $\mathtt{G}$, $\mathtt{P}$, $\mathtt{H}$ are defined in the usual way:

\begin{itemize}
	\item $\phi \vee \psi =_{def} \neg(\neg \phi \wedge \neg \psi)$,
	\item $\phi \rightarrow \psi =_{def} \neg \phi \vee \psi$,
	\item $\phi \leftrightarrow \psi =_{def} (\phi \rightarrow \psi) \wedge (\psi \rightarrow \phi)$,
	\item $\mathtt{F} \psi =_{def} (\psi \rightarrow \psi) \mathtt{U} \psi$,
	\item $\mathtt{G} \psi =_{def} \neg \mathtt{F} \neg \psi$,
	\item $\mathtt{P} \psi =_{def} (\psi \rightarrow \psi) \mathtt{S} \psi$,
	\item $\mathtt{H} \psi =_{def} \neg \mathtt{P} \neg \psi$,
	\item $\bigcirc^0 \psi =_{def} \psi; \bigcirc^{n+1} \psi=\bigcirc \bigcirc^{n} \psi, n \geqslant 0$,
	\item $\CIRCLE^0 \psi =_{def} \psi; \CIRCLE^{n+1} \psi=\CIRCLE \CIRCLE^{n} \psi, n \geqslant 0$.
\end{itemize}
%$\mathtt{F}$ ne ukljucuje $t_0$

Nonempty sets of formulas will be called theories.

In this paper we will consider time flow which is isomorphic to the set $\mathbb{N}$. We will take into account both future and past. Since we are dealing with a multi-agent system were agents have to share knowledge among them the obvious choice is to use the logic of time and knowledge, similarly like in \cite{halpern}.

We define $\Phi_{k}(\tau, (\theta_j)_{j\in \mathbb{N}})$ as a $k$-nested implication for the knowledge of an
agent $i$ and for formula $\tau$ based on the sequence of formulas $(\theta_j)_{j\in \mathbb{N}}$ in the following recursive way:

$$\Phi_{0}(\tau, (\theta_j)_{j\in \mathbb{N}}) = \theta_0 \rightarrow \tau,$$ 
$$\Phi_{k+1}(\tau, (\theta_j)_{j\in \mathbb{N}}) = \theta_{k+1} \rightarrow \mathtt{K}_i\Phi_{k}(\tau, (\theta_j)_{j\in \mathbb{N}}), \text{ for some } 0\leqslant i<m.$$

For example, $\Phi_{3}(\tau, (\theta_j)_{j\in \mathbb{N}}) = \theta_{3} \rightarrow \mathtt{K}_i(\theta_{2} \rightarrow \mathtt{K}_j(\theta_{1} \rightarrow \mathtt{K}_i(\theta_0 \rightarrow \tau))), 0\leqslant i,j <m$. This
definition follows the form of probabilistic k-nested implication presented in \cite{milos, probknow}.

\subsection{Semantics}
\noindent

We will defined models as Kripke's structures.

\begin{definition}
	
	A model $\mathcal{M}$ is any tuple $\langle R, W, \pi, \mathcal{K}\rangle$ such that
	\begin{itemize}
		\item Set of all possible runs $R$:
		\begin{itemize}
			\item $r^j = \{\langle x_{0}^{j,t}, \ldots , x_{m-1}^{j,t}\rangle| t=0,1,2\ldots\}, x_i^{j,t} \in \{\top,\bot\}$,
			\item $R=\{r^j; j=0,1,\ldots \}$,
			\item Restriction: if $x_{i}^{j,t}=\top$ then $x_{i}^{j,t+1}=\top$
		\end{itemize}
		\item $W$ set of time instances (the time flow isomorphic to $\mathbb{N}$),
		\item $\pi: R \times  W \times \mathbf{N_1} \rightarrow \{\top,\bot\}$ truth assignment, such that:
		\begin{itemize}
			\item $\pi(r^j, t, n_l) = x_l^{j,t}$,
			\item $\pi(r^j, t, u) = \top$
		\end{itemize}
		\item $\mathcal{K}$ possibility relations: $\mathcal{K}_i \subset (R \times W)^2$: $\langle r^j, t \rangle \mathcal{K}_i \langle r^{j'}, t' \rangle$ iff $x_{i}^{j,t}=x_{i}^{j',t'}$.
	\end{itemize}
\end{definition}

	\begin{figure}[h]
	\begin{center}
		\includegraphics[width=\textwidth]{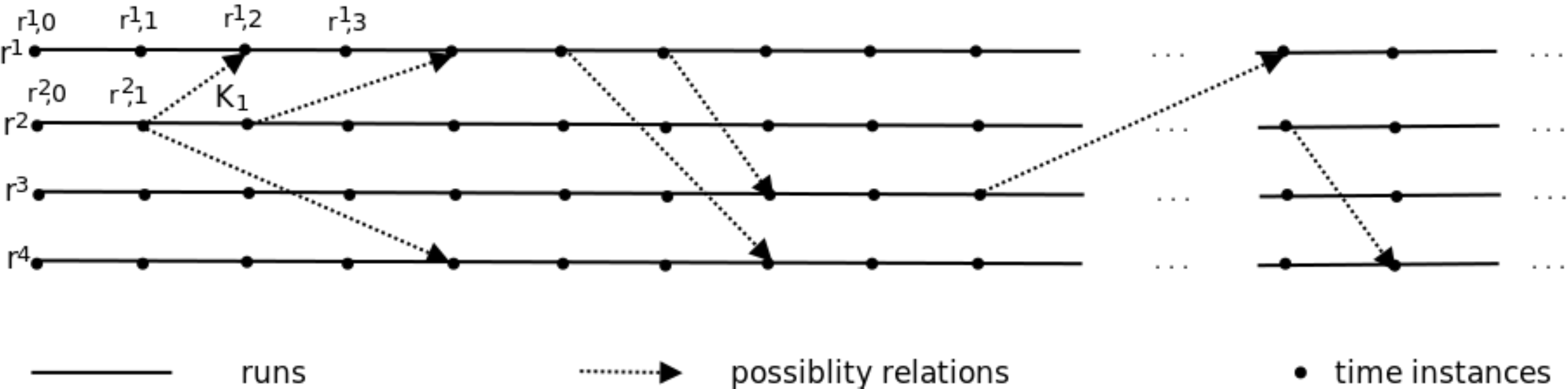}
		\caption{Kripke model}\label{slika_modela}
		\end{center}
	\end{figure}

Figure \ref{slika_modela} illustrates a Kripke model which contains the runs $r^1, r^2, r^3, r^4$, where $r^1$ is the sequence of $\langle r^1,0\rangle$, $\langle r^1,1\rangle$, $\langle r^1,2\rangle$, etc. and similarly for other runs. In this model, for example $\langle r^2, 1 \rangle \mathcal{K}_1 \langle r^{2}, 2 \rangle$, etc.

An $n_i \in \mathbf{N}$ is $true$ in the time instant $t$ in the run $r^j$ $(x_{i}^{j,t} = \top)$ if the Chord network node $i$ is active in the corresponding realization of the network. We define the set of propositional variables which represent the active nodes of Chord network as $\mathbf{N_a}= \{n_i | n_i \text{ is } true\}$.
For $n_i, n_j, n_k \in \mathbf{N}$ we define the relation $\mathtt{M}$ which represents the fact that $n_i$ is the member of the ring interval $(n_j, n_k]$ as: $n_i \mathtt{M} \langle n_j, n_k \rangle $ is $true$ iff 
\begin{itemize}
	\item $j=k$,
	\item $j<k$ and $j < i \leqslant k$,
	\item $k<j$ and $\neg (k <i \leqslant j)$.
\end{itemize}

\subsection{Satisfiability relation}

A formula is satisfiable if it is possible to find an interpretation, i.e. model, that makes that formula true.

\begin{definition}
	Let $\mathcal{M} = \langle R, W, \pi, \mathcal{K}\rangle$ be any model. The satisfiability relation $\models$ (formula $\alpha$ is satisfied in a time instance of a run $R \times W \models \alpha$) is defined recursively as follows:
\begin{enumerate}
	\item $\langle r^j, t \rangle \models n$ iff  $\pi(r^j, t, n)=true$, $n \in \mathbf{N_1}$
	
	\item $\langle r^j, t \rangle\models \alpha \wedge \beta$ iff  $\langle r^j, t \rangle \models \alpha$ and $\langle r^j, t \rangle \models \beta$
	\item $\langle r^j, t \rangle \models \neg \alpha$ iff  not $\langle r^j, t \rangle \models \alpha$ ( $\langle r^j, t \rangle \not\models \alpha$)
	\item $\langle r^j, t \rangle \models \bigcirc \alpha$ iff $\langle r^j, t+1 \rangle \models \alpha$
	\item $\langle r^j, t+1 \rangle \models \CIRCLE \alpha$ iff $\langle r^j, t \rangle \models \alpha$
	\item $\langle r^j, 0 \rangle \models \CIRCLE \alpha$
	\item $\langle r^j, t \rangle \models \alpha \mathtt{U} \beta$ iff there is a $i\geqslant 0$ such that $\langle r^j, t+i \rangle \models \beta$, and for every $k$, such that $0\leqslant k<i$, $\langle r^j, t+k \rangle \models \alpha$
	\item $\langle r^j, t \rangle \models \alpha \mathtt{S} \beta$ iff there is a $0 \leqslant i \leqslant t$ such that $\langle r^j, t-i \rangle \models \beta$, and for every $k$, such that $0\leqslant k<i$, $\langle r^j, t-k \rangle \models \alpha$
	\item $\langle r^j, t \rangle \models \mathtt{K}_i \alpha$ iff  $\langle r^{j'}, t' \rangle \models \alpha$ for all $\langle r^{j'}, t' \rangle \in \mathcal{K}_i(\langle r^j, t \rangle)$
	%\item $p \models \mathtt{E}_G \alpha$ iff  $p \models \mathtt{K}_i \alpha$ for^j all $i \in G$
	%\item $p \models \mathtt{C}_G \alpha$ iff  $p \models (\mathtt{E}_G)^k \alpha$ for^j all $k \in \mathbb{N}, k<m$
	%\item $\langle r^j, t \rangle \models \mathtt{D}_G \alpha$ iff  $\langle r^{j'}, t' \rangle \models \alpha$ for all $\langle r^{j'}, t' \rangle \in \cap_{i \in G} \mathcal{K}_i(\langle r^j, t \rangle)$  
	\item $\langle r^j, t \rangle \models n_i \succ n_j$ iff
	\begin{enumerate}
		\item $i=j$ and $\langle r^j, t \rangle \models n_i \wedge \mathtt{K}_i (\bigwedge_{n_j \in \mathbf{N}\backslash \{n_i\}} \neg n_j)$

		\item $i < j \leqslant m$ and $\langle r^j, t \rangle \models n_i \wedge n_j \wedge \mathtt{K}_i (\bigwedge_{k=i+1}^{j-1} \neg n_k) \wedge \mathtt{K}_i n_j$ 
		
		\item $j < i < m$ and $\langle r^j, t \rangle \models n_i \wedge n_j \wedge \mathtt{K}_i(\bigwedge_{k=i+1}^{m}  \neg n_k) \wedge \mathtt{K}_i (\bigwedge_{k=1}^{j-1} \neg n_k) \wedge \mathtt{K}_i n_j$  
		
		\item $j<i$ and $i = m$ and $\langle r^j, t \rangle \models n_i \wedge n_j \wedge \mathtt{K}_i (\bigwedge_{k=1}^{j-1} \neg n_k) \wedge \mathtt{K}_i n_j$  	  
		
	\end{enumerate} 
	\item $\langle r^j, t \rangle \models n_j \prec n_i$ iff
	\begin{enumerate}
		\item $i=j$, $t \neq 0$ and $\langle r^j, t \rangle \models n_i \wedge \mathtt{K}_i( \bigwedge_{n_k \in \mathbf{N}\backslash\{n_i\}} \neg n_k)$ 
		\item $i < j \leqslant m$ and $\langle r^j, t \rangle \models n_i \wedge n_j \wedge \mathtt{K}_i (\bigwedge_{k=i+1}^{j-1} \neg n_k) \wedge \mathtt{K}_i n_j$ 
		
		\item $j < i < m$ and $\langle r^j, t \rangle \models n_i \wedge n_j \wedge \mathtt{K}_i(\bigwedge_{k=i+1}^{m}  \neg n_k) \wedge \mathtt{K}_i (\bigwedge_{k=1}^{j-1} \neg n_k) \wedge \mathtt{K}_i n_j$  
		
		\item $j<i$ and $i = m$ and $\langle r^j, t \rangle \models n_i \wedge n_j \wedge \mathtt{K}_i (\bigwedge_{k=1}^{j-1} \neg n_k) \wedge \mathtt{K}_i n_j$ 
		\item $n_i=u$ and  $\langle r^j, t \rangle \models \neg n_j \vee (n_j \wedge (\CIRCLE (\neg \mathtt{K}_k (n_k \succ n_j))))$
		
	\end{enumerate}
\end{enumerate}
\end{definition}

\subsection{Axiomatization}
\noindent

The axioms of our theory are all instances of the following schemata:

\begin{enumerate}
	\item[A1] instances of tautologies
	
	\item[AT1] $\neg \bigcirc \alpha \leftrightarrow \bigcirc \neg \alpha$
	\item[AT2] $\bigcirc (\alpha \rightarrow \beta) \rightarrow (\bigcirc \alpha \rightarrow \bigcirc \beta)$ 
	\item[AT3] $\alpha \mathtt{U} \beta \leftrightarrow \beta \vee (\alpha \wedge \bigcirc (\alpha \mathtt{U} \beta))$
	\item[AT4] $\alpha \mathtt{U} \beta \rightarrow \mathtt{F} \beta$
	
	\item[AT5] $\CIRCLE (\alpha \rightarrow \beta) \rightarrow (\CIRCLE \alpha \rightarrow \CIRCLE \beta)$ 
	\item[AT6] $\alpha \mathtt{S} \beta \leftrightarrow (\CIRCLE \bot \wedge \beta) \vee (\neg \CIRCLE \bot \wedge (\beta \vee (\alpha \wedge \CIRCLE (\alpha \mathtt{S} \beta))))$
	\item[AT7] $\alpha \mathtt{S} \beta \rightarrow \mathtt{P} \beta$
	\item[AT8] $\alpha \leftrightarrow \bigcirc \CIRCLE \alpha$
	\item[AT9] $\alpha \leftrightarrow \CIRCLE \bot \vee (\CIRCLE \neg \bot \wedge \CIRCLE \bigcirc \alpha)$
	\item[AT10] $\mathtt{P} \bot$
	
	\item[AT11] $n_i \rightarrow \mathtt{G} n_i$
	
	\item[AK1] $\varphi \leftrightarrow \mathtt{K}_i \varphi$, $\varphi = \pm n_i$
	
	\item[AK2] $(\mathtt{K}_i \alpha \wedge \mathtt{K}_i (\alpha \rightarrow \beta)) \rightarrow \mathtt{K}_i \beta$
	\item[AK3] $\mathtt{K}_i \alpha \rightarrow \alpha$ 
	\item[AK4] $\mathtt{K}_i \alpha \rightarrow \mathtt{K}_i \mathtt{K}_i \alpha$
	\item[AK5] $\neg \mathtt{K}_i \alpha \rightarrow \mathtt{K}_i \neg \mathtt{K}_i \alpha$
	%\item[AK6] $\mathtt{K}_i \alpha \rightarrow  \mathtt{D}_G \alpha, i \in G$
	%\item[AK7] $(\mathtt{D}_G \alpha \wedge \mathtt{D}_G (\alpha \rightarrow \beta)) \rightarrow \mathtt{D}_G \beta$
	%\item[AK8] $\mathtt{D}_G \alpha \rightarrow \alpha$ 
	%\item[AK9] $\mathtt{D}_G \alpha \rightarrow \mathtt{D}_G \mathtt{D}_G \alpha$
	%\item[AK10] $\neg \mathtt{D}_G \alpha \rightarrow \mathtt{D}_G \neg \mathtt{D}_G \alpha$
	
	\item[AS1] $n_i \succ n_j \rightarrow \bigwedge_{n_k \in \mathbf{N_1} \backslash \{n_j\}}\neg (n_i  \succ n_k), n_i, n_j \in \mathbf{N}$ 
	\item[AS2] $n_i \prec n_j \rightarrow \bigwedge_{n_k \in \mathbf{N_1} \backslash \{n_j\}}\neg (n_i  \prec n_k), n_i, n_j \in \mathbf{N}$
	\item[AS3] $n_i \prec n_j \rightarrow \bigwedge_{n_k \in \mathbf{N_1} \backslash \{n_i\}}\neg (n_k  \prec n_j), n_i, n_j \in \mathbf{N}$
	\item[AS4] $n_i \prec n_j \rightarrow n_j \succ n_i, n_i, n_j \in \mathbf{N}$
	\item[AS5] $n_i \succ n_j \rightarrow \mathtt{K}_i(n_i \succ n_j), n_i, n_j \in \mathbf{N}$
	%\item[AS6] $(n_i \prec n_j) \wedge \bigwedge_{n_k \in \mathbf{N}} \pm n_k \wedge  \bigwedge_{n_k \in \mathbf{N}}(n_k \rightarrow \bigvee_{n_l \in \mathbf{N}}\mathtt{K}_{k}  (n_{k} \succ n_{l}) )$ \\$\rightarrow  \mathtt{D}_{\mathbf{N'}} \underline{\bigvee}_{k=1}^m (n_i \succ^k n_j)$, where $\mathbf{N'} =\{n_k: n_k \text{ appears in } \bigwedge_{n_l \in \mathbf{N} }\pm n_l\},$ $n_i \in \mathbf{N},n_j \in \mathbf{N}$
	\item[AS6] $((n_i \succ n_j) \wedge n_k \mathtt{M} \langle n_i, n_j \rangle \wedge \bigcirc(\neg \mathtt{K}_i n_k)) \rightarrow \bigcirc(n_i \succ n_j), n_i, n_j, n_k \in \mathbf{N}$ 
\end{enumerate}

Inference rules:
\begin{enumerate}
	\item[MP] from $\alpha$ and $\alpha \rightarrow \beta$ infer $\beta$
	\item[RTN] from $\alpha$ infer $\bigcirc \alpha$
	\item[RKN] from $\alpha$ infer $\mathtt{K}_i \alpha$
	\item[RI] from $\Phi_{k}(\neg((\bigwedge_{l=0}^i \bigcirc^l \alpha) \wedge \bigcirc^{i+1} \beta), (\theta_j)_{j\in \mathbb{N}})$ for all $i \geqslant 0$ infer $\Phi_{k}(\neg (\alpha \mathtt{U} \beta), (\theta_j)_{j\in \mathbb{N}})$
\end{enumerate}

[AT1 -- AT10] are standard axioms of the linear temporal logic. [AT11] takes into consideration specificity of our model and the restriction that when some $n_i$ become $\top$, then it will never be $\bot$. While [AK1] takes into consideration specificity of our model, [AK2 -- AK5] are standard axioms for reasoning about knowledge.
[AS1] says that a node can have only one successor. 
[AS2] says that a node can be predecessor of only one node.
[AS3] says that a node can have only one predecessor.
[AS4] says that if a node is predecessor of some other node, that other node has to be its successor.
[AS5] says that if a node $n_i$ has the successor $n_j$ than it knows that $n_j$ is its successor.
%[AS6] says that if a node $n_j$ is predecessor to $n_i$, then there is a "chain" of nodes where $n_i$ is at the beginning, $n_j$ at the end, and the next member of the "chain" is the successor of the previous one. 
[AS6] says when the current successor will be the successor in the next time instance.

[MP] is modus ponens, [RTN] and [RKN] resemble necessitations,and [RI] is the infinitary inference rule that characterize the Until operator.

\subsection{Soundness, Completeness and Decidability}
\noindent

In this part we will prove that our system is sound, complete and decidable. Informally speaking, the soundness means that we cannot prove anything that is wrong, the completeness means that we can prove everything that is right, and the decidability means that there is an effective method for determining whether arbitrary formula is a theorem of our logical system.

The inference relation $\vdash$ is defined as follows:

\begin{definition}
	We say that $\alpha$ is syntactical
	consequence of a set of formulas $T$ (or that $\alpha$ is deducible or derivable from $T$) and write $T\vdash \alpha$ iff there exists an at
	most countably infinite sequence of $\alpha_0$, \ldots , $\alpha_{\phi}$ such that $\alpha_{\phi}$ = $\alpha$ and for all $\beta \leqslant \phi$, $\alpha_{\beta}$ is an instance of some axiom, $\alpha_{\beta} \in T$, or $\alpha_{\beta}$ can be obtained from some previous members of the sequence
	by an application of some inference rule. A formula $\alpha$ is a theorem ($\vdash \alpha$) if it is deducible from the empty set. The rules {\em [RTN]} and {\em [RKN]}  can be applied only to theorems.
\end{definition}

\begin{definition}
	A set T is inconsistent iff $T \vdash \bot$, otherwise it is consistent. A set $T$ of formulae is maximal if for every formula $\alpha$ either $\alpha \in T$ or $\neg \alpha \in T$. A set $T$ is deductively closed if for
	every formula $\alpha$, if $T \vdash \alpha$, then $\alpha \in T$.
\end{definition}

%\begin{definition}
%	For a given set of formulae $T$ we define set $\ast T =\{\ast \alpha | \alpha \in T\}$ and $\ast \in \{\mathtt{K}_i, \bigcirc, \CIRCLE\}, 0 \leqslant i < m$. Also, For a given set of formulae $T$ we define set $\mathtt{K}_i^{-}(T)=\{\alpha|\mathtt{K}_i \alpha \in T\}$.
%\end{definition}

\begin{restatable}{theorem}{soundness}[Soundness]\label{soundness}
	$ \vdash \alpha$ implies $ \models \alpha$.
\end{restatable}

\begin{restatable}{theorem}{maxext}\label{max_ext}
Every consistent set of formulas $T$ can be extended to a maximal consistent set $T^*$.
\end{restatable}

\noindent
\textbf{Canonical structure}. We define a special, so called canonical structure $\mathbb{M^*} = \langle R, W, \pi, \mathcal{K}\rangle$. Let $\mathcal{T}$ be the set of all maximal consistent sets. Let
$T \in \mathcal{T}$. We define a run inductively as: $T^0_{j} = T$, and $T^t_{j} = \{ \alpha : \bigcirc \alpha \in T^{t-1}_{j} \}, t>0$.

We denote:
\begin{itemize}
	\item $r^j_t = \langle x^{j,t}_0, \ldots, x^{j,t}_{m-1}\rangle$, where $x^{j,t}_l = \top$ if $n_l \in T^t_j$, and $x^{j,t}_l = \bot$ otherwise,
	\item $r^j = \langle x^{j,t}_0, \ldots, x^{j,t}_{m-1}\rangle, t = 0, 1 \ldots$,
	\item $R=\{r^j\}$.
\end{itemize}

Also:
\begin{itemize}
	\item $\pi (r^j, t, n_l) = x_l^{j,t}$,
	\item $\langle r^j, t \rangle \mathcal{K}_i \langle r^{j'}, t' \rangle$ iff $n_i \in T^t_j \Leftrightarrow n_i \in T^{t'}_{j'}$.
\end{itemize}

\begin{restatable}{theorem}{strongcompletness}[Strong completeness]\label{completness}
	Every consistent set of formulas is satisfiable.
\end{restatable}

\begin{theorem}\label{completness1}
	$T \models \psi \leftrightarrow T \vdash \psi$.
\end{theorem}

\begin{theorem}[Decidability theorem]\label{decidabitlity}
	Checking the satisfiability of a given formula $\psi$ is decidable.
	
	\begin{proof}
	In every run at some time instance we will have stationary situation (at least when all possible nodes join the system). Since we do not allow leaving of the nodes, we can apply the ideas from \cite{omegakv, zorano2006, branchtime} to prove the decidability problem.
	\end{proof}
\end{theorem}

\section{Proof of the Correctness}\label{proof}

To be able to prove the correctness of the Chord protocol we need to introduce the following definitions:

\begin{definition}[Stable pair]\label{stablepair}
	 The pair of nodes $\langle n_k , n_l \rangle$  is stable (we denote it with $n_k \Cap n_l$) at $\langle r^{j'},t \rangle$  iff
	($n_l \succ^{m_1} n_k) \wedge (\bigwedge_{1 \leqslant j \leqslant m_1} \mathtt{K}_{i_j}(n_{i_j} \succ n_{i_{j+1}})) \wedge (n_l \prec^{m_1} n_k) \wedge (\bigwedge_{1 \leqslant j \leqslant m_1}  \mathtt{K}_{i_{j+1}} (n_{i_{j+1}} \prec n_{i_j}))$, where $n_{i_j}\in \mathbf{N_a}$.
\end{definition}

\begin{definition}[Stable network]\label{stablenet}
	Network is stable (we denote it with $\circledcirc$) at $\langle r^j,t \rangle$ iff $n_k \Cap n_k$ for all $n_k \in \mathbf{N_a}$.
\end{definition}

We introduce an integer constant $f \in \mathbb{N}$, that will represent fairness condition, i.e. it guarantees that a formula will be realized at maximum of $f$ time instances. 

The processes of the Chord network can be describe with:
\begin{itemize}
	\item[] $\rho_{S}$: $\mathtt{H} (\bigwedge_{n_j \in \mathbf{N}} \neg n_j) \wedge n_i \wedge (\bigwedge_{n_j \in \mathbf{N} \backslash \{n_i\}} \neg n_j)\wedge \mathtt{K}_i (n_i \succ n_i) \wedge \mathtt{K}_i (n_i \prec u)$ for one $n_i \in \mathbf{N}$,
	\item[] $\rho_{J,i}$: $\CIRCLE (\neg n_i) \wedge
	n_i \wedge 
	\bigvee_{l=0}^f \bigcirc^l \mathtt{K}_i (n_i \succ n_j) \wedge \mathtt{K}_i (n_i \prec u)$, $n_j \in \mathbf{N_a}, n_i \in \mathbf{N}, i \neq j$,
	\item[] $\rho_{S1,i,j}$: $(\mathtt{K}_i (n_i \succ n_{j}) \wedge \mathtt{K}_j (n_j \prec u)) \vee
	(\mathtt{K}_i (n_i \succ n_{j}) \wedge \mathtt{K}_j (n_j \prec n_k) \wedge n_i\mathtt{M}\langle n_k, n_{j}\rangle)
	\rightarrow \bigvee_{l=0}^f \bigcirc^l \mathtt{K}_j (n_j \prec n_{i})$, $n_i, n_k, n_j \in \mathbf{N_a}$,
	\item[] $\rho_{S2,i,j}$: $\mathtt{K}_i (n_i \succ n_{j}) \wedge \mathtt{K}_j (n_j \prec n_{k}) \wedge n_k\mathtt{M}\langle n_i, n_{j}\rangle \rightarrow $
	$\bigvee_{l=0}^f \bigcirc^l \mathtt{K}_i (n_i \succ n_k)$, $n_i, n_k, n_j \in \mathbf{N_a}$. 
\end{itemize}
[$\rho_{S}$] describes the start of the new Chord network.
[$\rho_{J,i}$] represents the situation when a new node $n_i$ is joining the existing Chord network, while [$\rho_{S1,i,j}$] and [$\rho_{S2,i,j}$] characterize stabilization processes in one Chord network.

To be able to describe periodicity of the stabilization process, we introduce the following axioms:

\begin{itemize}
	\item[] ACF1: $n_i \wedge \rho_{S} \rightarrow \bigvee_{l=0}^f \bigcirc^l \bigvee_{j=0}^{m-1} \rho_{S1,i,j}$, $n_i \in \mathbf{N_a}$,
	\item[] ACF2: $n_i \wedge \rho_{S} \rightarrow \bigvee_{l=0}^f \bigcirc^l  \bigvee_{j=0}^{m-1} \rho_{S2,i,j}$, $n_i \in \mathbf{N_a}$,
	\item[] ACF3: $n_i \wedge \rho_{J,i} \rightarrow \bigvee_{l=0}^f \bigcirc^l \bigvee_{j=0}^{m-1} \rho_{S1,i,j}$, $n_i \in \mathbf{N_a}$,
	\item[] ACF4: $n_i \wedge \rho_{J,i} \rightarrow \bigvee_{l=0}^f \bigcirc^l \bigvee_{j=0}^{m-1} \rho_{S2,i,j}$, $n_i \in \mathbf{N_a}$,
	\item[] ACF5: $n_i \wedge \rho_{S1,i,k} \rightarrow \bigvee_{l=0}^f \bigcirc^l \bigvee_{j=0}^{m-1}  \rho_{S1,i,j}$, $n_i \in \mathbf{N_a}, k\in \{0, m-1\}$,
	\item[] ACF6: $n_i \wedge \rho_{S2,i,k} \rightarrow \bigvee_{l=0}^f \bigcirc^l \bigvee_{j=0}^{m-1} \rho_{S2,i,j}$, $n_i \in \mathbf{N_a}, k\in \{0, m-1\}$.
\end{itemize}

The correctness of the Chord protocol can be proved by the following Lemmas and Theorem.

\begin{lemma}\label{start}
	Let a new node start a new Chord network. Then, there is a finite period of time after the network will be stable again, if no other nodes are trying to join in the meanwhile.
\end{lemma}

\begin{lemma}\label{join11}
	Let a new node join a stable Chord network which consists of only one node. Then, there is a
	finite period of time after the network will be stable again, if no other nodes are trying to join in the meanwhile.
\end{lemma}

Proofs of Lemmas \ref{start} and \ref{join11} are similar like the proof of Lemma \ref{join11n}. 

\begin{restatable}{lemma}{join}\label{join11n}
	Let a peer join a Chord network, between two nodes which constitute a stable pair, such that the second node is the successor of the first node. Then, there is a
	finite period of time after the starting pair will be stable again, if no other nodes are trying to join in the meanwhile.
\end{restatable}

\begin{lemma}\label{join1n}
	Let a peer join a Chord network, between two nodes which constitute a stable pair. Then, there is a
	finite period of time after the starting pair will be stable again, if no other nodes are trying to join in the meanwhile.
	\begin{proof}
		Since one new node is joining the stable pair, we can choose two nodes which are each others successor and predecessor and the new node is joining between them, so we can apply Lemma \ref{join11n}.
	\end{proof}
\end{lemma}

\begin{lemma}\label{joinnn}
	Let a Chord network contain a stable pair. If a sequence of nodes join between the nodes that constitute this stable pair, then   there is a finite period of time after the starting pair will be stable again.
	\begin{proof}
	If we assume that all nodes that want to join the network have different successors, by Lemma \ref{join1n} the statement holds.
	
	If this is not the case, we can assume that $n_i \Cap n_k$ and that set of nodes $j_1, j_2, \ldots$, such that $i \leqslant \ldots \leqslant j_2 \leqslant j_1 \leqslant k$, are joining this stable pair. Then, we can apply Lemma \ref{join1n} on the tuples $\langle i, j_1, k\rangle$, $\langle i, j_2, j_1\rangle$, \ldots. This process will have as a result $n_i \Cap n_k$, again.
	\end{proof}
\end{lemma}

\begin{theorem}\label{ultimatethm}
	$\vdash \neg \circledcirc \rightarrow \mathtt{F} \circledcirc$
	\begin{proof}
		The unstable state can be reached only by joining of the new nodes, and, since we do not allow node failures, this theorem is the corollary of Lemmas \ref{start} -- \ref{joinnn}.
	\end{proof}	
\end{theorem}

\section{Conclusion}\label{conclusion}

The core part of the every IoT system are its discovery and control services. In the distributed environment, these services can be realized using the Chord protocol.

In this paper we provide the axiomatization and prove the soundness, strong completeness and decidability of the logic of time and knowledge. Using this framework,  we prove the correctness of the maintenance of the ring topology of the Chord protocol with the respect of the fact that nodes are not allowed to departure the system after they join it.

Our plan is to continue our research to prove the correctness in the general case.
Also, one of the possible directions for further work is to apply the
similar technique to describe other DHT protocols and other cloud processes.

Another challenge could be to verify the given proof in one
of the formal proof assistants (e.g., Coq, Isabelle/HOL). It might
also produce a certified program implementation from the proof of
correctness. 

\appendix 
\section{Proofs}
\label{appendix}

%\setcounter{lemma}{0}
%\setcounter{theorem}{0}

%\begin{theorem}[Soundness]
%	$ \vdash \alpha$ implies $ \models \alpha$.
\soundness*	
	\begin{proof}

		\begin{itemize}
			\item[AT1] $\neg \bigcirc \alpha \leftrightarrow \bigcirc \neg \alpha$
			
			$\langle r^j, t\rangle \models \neg \bigcirc \alpha$ iff
			$\langle r^j, t\rangle \not\models \bigcirc \alpha$ iff
			
			$\langle r^j, t+1\rangle \not\models \alpha$ iff
			$\langle r^j, t+1\rangle \models \neg \alpha$ iff
			
			$\langle r^j, t\rangle \models \bigcirc \neg \alpha$
			
			\item[AT2] $\bigcirc (\alpha \rightarrow \beta) \rightarrow (\bigcirc \alpha \rightarrow \bigcirc \beta)$ 
			
			$\langle r^j, t\rangle \models \bigcirc (\alpha \rightarrow \beta)$ iff
			$\langle r^j, t+1\rangle \models \alpha \rightarrow \beta$ iff
			$\langle r^j, t+1\rangle \models \neg \alpha \vee \beta$ iff
			
			$\langle r^j, t+1\rangle \models \neg \alpha $ or $\langle r, t+1\rangle \models \beta$ iff
			
			$\langle r^j, t\rangle \models \neg \bigcirc \alpha $ or $\langle r, t\rangle \models \bigcirc \beta$ if
			
			$\langle r^j, t\rangle \models \neg \bigcirc \alpha \vee \bigcirc \beta$ iff
			
			$\langle r^j, t\rangle \models \bigcirc \alpha \rightarrow \bigcirc \beta$
			
			\item[AT3] $\alpha \mathtt{U} \beta \leftrightarrow \beta \vee (\alpha \wedge \bigcirc (\alpha \mathtt{U} \beta))$
			
			$\langle r^j, t\rangle \models \alpha \mathtt{U} \beta$ iff
			
			$\langle r^j, t+i\rangle \models \beta, i \geqslant 0$ and $\langle r^j, t+k\rangle \models \alpha$ and $0\leqslant k<i$ iff
			
			$\langle r^j, t\rangle \models \beta$ or
			$\langle r^j, t+i\rangle \models \beta, i > 0$ and $\langle r^j, t+k\rangle \models \alpha$ and $0\leqslant k<i$ iff
			
			$\langle r^j, t\rangle \models \beta$ or
			$\langle r^j, t\rangle \models \alpha$ and $\langle r, t+i\rangle \models \beta$ and $\langle r^j, t+k\rangle \models \alpha$ and $1\leqslant k<i$ iff
			
			$\langle r^j, t\rangle \models \beta$ or
			$\langle r^j, t\rangle \models \alpha$ and $\langle r, t+i-1\rangle \models \beta$ and $\langle r, t+k-1\rangle \models \alpha$ and $0\leqslant k-1<i-1$ iff
			
			$\langle r^j, t\rangle \models \beta$ or					 
			$\langle r^j, t\rangle \models \alpha$ and $\langle r^j, t\rangle \models \bigcirc (\alpha \mathtt{U} \beta)$ iff
			
			$\langle r^j, t\rangle \models \beta \vee (\alpha \wedge \bigcirc (\alpha \mathtt{U} \beta))$
			
			\item[AT4] $\alpha \mathtt{U} \beta \rightarrow \mathtt{F} \beta$
			
			$\langle r^j, t\rangle \models \alpha \mathtt{U} \beta$ iff
			
			$\langle r^j, t+i\rangle \models \beta, i \geqslant 0$ and $\langle r, t+k\rangle \models \alpha$ and $0\leqslant k<i$ if
			
			$\langle r^j, t+i\rangle \models \beta$ and exists $i \geqslant 0$ iff
			
			$\langle r^j, t\rangle \models \mathtt{F} \beta$
			
		\end{itemize}
		
		[AT5 -- AT7] similarly like [AT2 -- AT4] regarding the sub-case $\langle r^j, 0\rangle$.
		
		\begin{itemize}
			\item[AT8] $\alpha \leftrightarrow \bigcirc \CIRCLE \alpha$
			
			$\langle r^j, t\rangle \models \alpha $ iff
			$\langle r^j, t+1\rangle \models \CIRCLE \alpha $ iff 
			$\langle r^j, t\rangle \models \bigcirc \CIRCLE \alpha$
		\end{itemize}
		
		[AT9] similarly like [AT8], regarding the sub-case $\langle r^j, 0\rangle$.
		
		\begin{itemize}
			\item[AT10] $\mathtt{P} \bot$
			
			$\langle r^j, t\rangle \models \mathtt{P}\bot $ iff
			$\langle r^j, 0\rangle \models \CIRCLE \bot $ 
			
			\item[AT11] $n_i \rightarrow \mathtt{G} n_i$
			
			Because of the restriction if $x_{i}^{j,t}=\top$ then $x_{i}^{j,t+1}=\top$.
		\end{itemize}
		
		\begin{itemize}
			\item[AK1] $\varphi \leftrightarrow \mathtt{K}_i \varphi$, $\varphi = \pm n_{i}$
			
			$\langle r^j, t\rangle \models \varphi$ iff
			
			$\langle r^{j'}, t'\rangle \models \varphi$ and all $\langle r^{j'}, t'\rangle \in \mathcal{K}_i(\langle r^j, t\rangle)$ iff

			$\langle r^j, t\rangle \models \mathtt{K}_i \varphi$

			\item[AK2] $(\mathtt{K}_i \alpha \wedge \mathtt{K}_i (\alpha \rightarrow \beta)) \rightarrow \mathtt{K}_i \beta$
			
			$\langle r^j, t\rangle \models \mathtt{K}_i \alpha \wedge \mathtt{K}_i (\alpha \rightarrow \beta)$ iff
			
			$\langle r^j, t\rangle \models \mathtt{K}_i \alpha$ and $\langle r^j, t\rangle \models \mathtt{K}_i (\alpha \rightarrow \beta)$ iff
			
			$\langle r^{j'}, t'\rangle \models \alpha$ for all $\langle r^{j'}, t'\rangle \in \mathcal{K}_i(\langle r^j, t\rangle)$ and $\langle r', t'\rangle \models \alpha \rightarrow \beta$  for all $\langle r^{j'}, t'\rangle \in \mathcal{K}_i(\langle r^j, t\rangle)$ iff
			
			$\langle r^{j'}, t'\rangle \models \alpha \wedge (\alpha \rightarrow \beta)$  for all $\langle r^{j'}, t'\rangle \in \mathcal{K}_i(\langle r^j, t\rangle)$ if
			
			$\langle r^{j'}, t'\rangle \models \beta$  for all $\langle r^{j'}, t'\rangle \in \mathcal{K}_i(\langle r^j, t\rangle)$ iff
			
			$\langle r^j, t\rangle \models \mathtt{K}_i\beta$
			
			\item[AK3] $\mathtt{K}_i \alpha \rightarrow \alpha$ 
			
			$\langle r^j, t\rangle \models \mathtt{K}_i \alpha$ iff
			
			$\langle r^{j'}, t'\rangle \models \alpha$ and $\langle r^{j'}, t'\rangle \in \mathcal{K}_i(\langle r^j, t\rangle)$ then
			
			$\langle r^j, t\rangle \models \alpha$
			
			\item[AK4] $\mathtt{K}_i \alpha \rightarrow \mathtt{K}_i \mathtt{K}_i \alpha$
			
			$\langle r^j, t\rangle \models \mathtt{K}_i \alpha$ iff
			
			$\langle r^{j'}, t'\rangle \models \alpha$ and $\langle r^{j'}, t'\rangle \in \mathcal{K}_i(\langle r^j, t\rangle)$ then
			
			$\langle r^{j''}, t''\rangle \models \alpha$ and $\langle r^{j''}, t''\rangle \in \mathcal{K}_i(\langle r^{j'}, t'\rangle)$ iff 
			
			$\langle r^{j'}, t'\rangle \models \mathtt{K}_i \alpha$ and $\langle r^{j'}, t'\rangle \in \mathcal{K}_i(\langle r^j, t\rangle)$ iff
			
			$\langle r^j, t\rangle \models \mathtt{K}_i \mathtt{K}_i \alpha$
			
			\item[AK5] $\neg \mathtt{K}_i \alpha \rightarrow \mathtt{K}_i \neg \mathtt{K}_i \alpha$ 
			
			Similarly like [AK4]

			%\item[AK6] $\mathtt{K}_i \alpha \rightarrow  \mathtt{D}_G \alpha, i \in G$
			
			%$\langle r^j, t\rangle \models \mathtt{K}_i \alpha$ and $i \in G$ iff
			
			%	$\langle r^{j'}, t' \rangle \models \alpha$ for all $\langle r^{j'}, t'\rangle \in \mathcal{K}_i(\langle r^j, t\rangle)$  and $i \in G$ then
			
			%	$\langle r^{j'}, t'\rangle \models \alpha$ for all $\langle r^{j'}, t' \rangle \in \cap_{i \in G} \mathcal{K}_i(\langle r^j, t \rangle)$ iff
			
			%	$\langle r, t\rangle \models  \mathtt{D}_G \alpha$
		\end{itemize}
		
		%	[AK7 -- AK10] similarly like [AK2 -- AK5]
		
		\begin{itemize}
			\item[AS1] $n_i \succ n_j \rightarrow \bigwedge_{n_k \in \mathbf{N_1} \backslash \{n_j\}}\neg (n_i  \succ n_k), n_i, n_j \in \mathbf{N}$ 
			
			Let $\langle r^l, t\rangle \models n_i \succ n_j$. 
			
			If $i=j$ , by the definition of $n_i \succ n_i$ we have that 
			
			\begin{center} 			
				$\langle r^l, t \rangle \models n_i \wedge \mathtt{K}_i (\bigwedge_{n_j \in \mathbf{N}\backslash \{n_i\}} \neg n_j),$
			\end{center} 
			
			and by [AK3] and [A1] we have that  
			
			\begin{center}
				$\langle r^l, t \rangle \models \bigwedge_{n_j \in \mathbf{N}\backslash \{n_i\}} \neg n_j$ 	
			\end{center}
			
			so, there is no candidate $n_k \in \mathbf{N_1}  \backslash \{n_i\}$ such that $n_i  \succ n_k$.	
			
			Let $i \neq j$ and let $$\langle r^l, t\rangle \models n_i \succ n_{j'}$$ and $i < j' < j \leqslant m$. Then, by the definition of $\succ$ relation we have
			
			$$n_i \succ n_{j'} \rightarrow \mathtt{K}_i n_{j'},$$
			
			and
			
			$$n_i \succ n_{j} \rightarrow \mathtt{K}_i (\bigwedge_{k=i+1}^{j-1} \neg n_k).$$
			
			Last two facts are in contradiction.
			
			Similarly in all other cases.
			
			\item[AS2] $n_i \prec n_j \rightarrow \bigwedge_{n_k \in \mathbf{N_1} \backslash \{n_j\}}\neg (n_i  \prec n_k), n_i, n_j \in \mathbf{N}$
			
			Similarly like [AS1].
			
			\item[AS3] $n_i \prec n_j \rightarrow \bigwedge_{n_k \in \mathbf{N_1} \backslash \{n_i\}}\neg (n_k  \prec n_j), n_i, n_j \in \mathbf{N}$
			
			Similarly like [AS1].

			\item[AS4] $n_i \prec n_j \rightarrow n_j \succ n_i, n_i, n_j \in \mathbf{N}$ 
			
			Let $\langle r^l, t\rangle \models n_i \prec n_j$. 
			
			If $i=j$ , by the definition of $n_i \prec n_i$ we have that 
			
			\begin{center} 			
				$\langle r^l, t \rangle \models n_i \wedge \mathtt{K}_i (\bigwedge_{n_j \in \mathbf{N}\backslash \{n_i\}} \neg n_j),$
			\end{center}
			
			so we have $\langle r^l, t\rangle \models n_i \succ n_j$. 
			
			Let $i \neq j$ and  $j < i \leqslant m$, then 
			$$\langle r^l, t\rangle \models  n_j \wedge n_i \wedge \mathtt{K}_j (\bigwedge_{k=j+1}^{i-1} \neg n_k) \wedge \mathtt{K}_j n_i,$$
			
			which means that: $$\langle r^l, t\rangle \models n_j \succ n_i.$$ Similarly in all other cases.
			
			\item[AS5] $n_i \succ n_j \rightarrow \mathtt{K}_i(n_i \succ n_j), n_i, n_j \in \mathbf{N}$
			
			Suppose opposite, that 
			$$\langle r^l, t\rangle \models (n_i \succ n_{j}) \wedge \neg \mathtt{K}_i (n_i \succ n_{j}) .$$
			
			Then we have:
			$$\langle r^l, t\rangle \models n_i \succ n_{j} \text{ and } \langle r^l, t\rangle \models \neg \mathtt{K}_i (n_i \succ n_{j}) ),$$
			
			and 
			$$\langle r^l, t\rangle \not\models \mathtt{K}_i (n_i \succ n_{j}).$$
			
			By [AK3], we have that:
			$$\langle r^l, t\rangle \not\models n_i \succ n_{j},$$
			
			which cannot hold.

			\item[AS6] $((n_i \succ n_j) \wedge n_k \mathtt{M} \langle n_i, n_j \rangle \wedge \bigcirc(\neg \mathtt{K}_i n_k)) \rightarrow \bigcirc(n_i \succ n_j), n_i, n_j, n_k \in \mathbf{N}$ 
			
			Let $\langle r^{j'}, t\rangle \models ((n_i \succ n_j) \wedge n_k \mathtt{M} \langle n_i, n_j \rangle \wedge \bigcirc(\neg \mathtt{K}_i n_k))$, and suppose that $\langle r^{j'}, t+1\rangle \not\models (n_i \succ n_j)$. Without loss of generality we can assume that $j < i \leqslant m-1$. The, by the definition of $\succ$ and $\mathtt{M}$ we have that 
			\begin{equation}\label{korak}
			\langle r^{j'}, t\rangle \models n_i \wedge n_j \wedge \mathtt{K}_i (\bigwedge_{l=i+1}^{j-1} \neg n_l) \wedge \bigcirc \mathtt{K}_i (\bigwedge_{l=i+1}^{j-1} \neg n_l) \wedge \mathtt{K}_i n_j.
			\end{equation}

			From  $\langle r^{j'}, t+1\rangle \not\models (n_i \succ n_j)$ we can conclude that for some $k, i<k<j$ $\langle r^{j'}, t+1\rangle \models (n_i \succ n_k)$. By the definition of $\succ$: 
			$$\langle r^{j'}, t+1\rangle \models n_i \wedge n_k \wedge \mathtt{K}_i (\bigwedge_{l=i+1}^{k-1} \neg n_l) \wedge \mathtt{K}_i n_k,$$
			
			which is in contradiction with (\ref{korak}).
			\end{itemize}
			
			[MP], [RTN] and [RKN] in standard way (see \cite{omegakv, zorano2006, branchtime})
			
			\begin{itemize}
			\item[RI] from $\Phi_{k}(\neg((\bigwedge_{l=0}^i \bigcirc^l \alpha) \wedge \bigcirc^{i+1} \beta), (\theta_j)_{j\in \mathbb{N}})$ for all $i \geqslant 0$ infer $\Phi_{k}(\neg (\alpha \mathtt{U} \beta), (\theta_j)_{j\in \mathbb{N}})$
			
			We show that [RI] produces valid formula for a valid set of premises by induction on $k$. Suppose that
			$$\langle r^{j'}, t\rangle \models \Phi_{k}(\neg((\bigwedge_{l=0}^i \bigcirc^l \alpha) \wedge \bigcirc^{i+1} \beta), (\theta_j)_{j\in \mathbb{N}}), \text{ for } i \geqslant 0.$$ Then $$\langle r^{j'}, t\rangle \models \Phi_{k}(\neg (\alpha \mathtt{U} \beta), (\theta_j)_{j\in \mathbb{N}})$$ by following:
			Induction base. 
			
			$k=0$:
			
			Note that:
			
			if $\langle r^{j'}, t\rangle \not\models \theta_0 \rightarrow \neg (\alpha \mathtt{U} \beta)$ then
						
			$\langle r^{j'}, t\rangle \models \theta_0 \wedge (\alpha \mathtt{U} \beta)$ iff
						
			$\langle r^{j'}, t\rangle \models \theta_0 $ and $\langle r^j, t\rangle \models \alpha \mathtt{U} \beta$ iff
			
			\begin{equation}
			\langle r^{j'}, t\rangle \models \theta_0 \text{ and }\langle r^{j'}, t+i_0\rangle \models \beta \text{ and } \langle r^{j'}, t+l\rangle \models \alpha, 0\leqslant l<i_0 \label{pomoc}
			\end{equation}

			$\langle r^{j'}, t\rangle \models \{\theta_0 \rightarrow \neg((\bigwedge_{l=0}^i \bigcirc^l \alpha) \wedge \bigcirc^{i+1} \beta) | i \geqslant 0\}$ iff
			
			\begin{equation}
			\langle r^{j'}, t\rangle \models \neg \theta_0 \text{ and }(\langle r^{j'}, t+i\rangle \not\models \beta \text{ and }\langle r^{j'}, t+l\rangle \not\models \alpha, 0\leqslant l<i \text{ for all } i \geqslant 0  
			\end{equation}
			
			which in contradiction with (\ref{pomoc}).

			Inductive step. 
			
			Let $\langle r^{j'}, t\rangle \models \Phi_{k+1}(\neg((\bigwedge_{l=0}^i \bigcirc^l \alpha) \wedge \bigcirc^{i+1} \beta), (\theta_j)_{j\in \mathbb{N}})$ for $i \geqslant 0$, i.e.
			$$\langle r^{j'}, t\rangle \models \theta_{k+1} \rightarrow \mathtt{K}_{e_k} \Phi_{k}(\neg((\bigwedge_{l=0}^i \bigcirc^l \alpha) \wedge \bigcirc^{i+1} \beta), (\theta_j)_{j\in \mathbb{N}})$$ for $i \geqslant 0$ and $0 \leqslant e_k <m$.
			Let us assume opposite, that 
			$$\langle r^{j'}, t\rangle \not\models \Phi_{k+1}(\neg (\alpha \mathtt{U} \beta), (\theta_j)_{j\in \mathbb{N}}), \text{ i.e.}$$ 
			$$\langle r^{j'}, t\rangle \models \theta_{k+1} \wedge \neg \Phi_{k}(\neg (\alpha \mathtt{U} \beta), (\theta_j)_{j\in \mathbb{N}}).$$
			Also, we have:
			$$\langle r^{j'}, t\rangle \models \mathtt{K}_{e_k} \Phi_{k}(\neg((\bigwedge_{l=0}^i \bigcirc^l \alpha) \wedge \bigcirc^{i+1} \beta), (\theta_j)_{j\in \mathbb{N}})$$ for $i \geqslant 0$. For every $\langle r^{j''}, t'\rangle \in \mathcal{K}_{e_k}(\langle r^{j'}, t\rangle)$ we have that:
			$$\langle r^{j''}, t'\rangle \models \Phi_{k}(\neg((\bigwedge_{l=0}^i \bigcirc^l \alpha) \wedge \bigcirc^{i+1} \beta), (\theta_j)_{j\in \mathbb{N}})$$ and by induction hypothesis 
			$$\langle r^{j''}, t'\rangle \models \Phi_{k}(\neg (\alpha \mathtt{U} \beta), (\theta_j)_{j\in \mathbb{N}}).$$ Therefore:
			$$\langle r^{j'}, t\rangle \models \mathtt{K}_{e_k} \Phi_{k}(\neg (\alpha \mathtt{U} \beta), (\theta_j)_{j\in \mathbb{N}})$$ which is a contradiction.
			
		\end{itemize}
		
	\end{proof}
%\end{theorem}

\begin{theorem}[Deduction theorem]
	$T \cup \{\varphi\} \vdash \psi$ implies $T \vdash \varphi \rightarrow \psi$.
	\begin{proof}
		If $\psi$ is an axiom or $\psi \in T$, then $T \vdash \psi$, so since $T \vdash \psi \rightarrow (\varphi \rightarrow \psi)$ [A1] by [MP] $T \vdash \varphi \rightarrow \psi$. If $\varphi=\psi$ then $T \vdash \varphi \rightarrow \varphi$ [A1].
		
		If $\psi$ is a theorem then, $\vdash \bigcirc \psi$. By weakening $T \vdash \bigcirc \psi$, so $T \vdash \varphi \rightarrow \bigcirc \psi$. Similarly for [RKN] rule.
		
		%If $T \cup \{\varphi\} \vdash \gamma \rightarrow \neg(\alpha \mathtt{U}\beta) $ is obtained by [RI] for $k=0$ then 
		%$$T \cup \{\varphi\} \vdash \gamma \rightarrow \neg((\bigwedge_{k=0}^i \bigcirc^k \alpha) \wedge \bigcirc^{i+1} \beta) \text{ for all } i \geqslant 0.$$ By induction hypothesis, we have 
		%$$T \vdash \varphi \rightarrow (\gamma \rightarrow \neg((\bigwedge_{k=0}^i \bigcirc^k \alpha) \wedge \bigcirc^{i+1} \beta))\text{ for all } i \geqslant 0,$$ and, also, 
		%$$T \vdash (\varphi \wedge \gamma) \rightarrow (\neg((\bigwedge_{k=0}^i \bigcirc^k \alpha) \wedge \bigcirc^{i+1} \beta))\text{ for all } i \geqslant 0 \text{ [A1].}$$ By applying [RI] we conclude 
		%$$T \vdash (\varphi \wedge \gamma) \rightarrow (\neg (\alpha \mathtt{U} \beta)),$$ and finally, we have $$T \vdash \varphi \rightarrow (\gamma \rightarrow \neg (\alpha \mathtt{U} \beta)) \text{ [A1].}$$
		
		Let us assume that $\psi$ if obtained from $T \cup \{\varphi\}$ using [RI] rule, i.e. $\psi = \Phi_{k}(\neg (\alpha \mathtt{U} \beta), (\theta_j)_{j\in \mathbb{N}})$. Then we have:
		\begin{itemize}
			\item[] $T, \varphi \vdash \Phi_{k}(\neg((\bigwedge_{l=0}^i \bigcirc^l \alpha) \wedge \bigcirc^{i+1} \beta), (\theta_j)_{j\in \mathbb{N}})$ for all $i \geqslant 0$,
			\item[] $T \vdash \varphi \rightarrow \Phi_{k}(\neg((\bigwedge_{l=0}^i \bigcirc^l \alpha) \wedge \bigcirc^{i+1} \beta), (\theta_j)_{j\in \mathbb{N}})$, by induction hypothesis,
			\item[] $T \vdash \varphi \rightarrow (\theta_k \rightarrow \mathtt{K}_{e_k} \Phi_{k-1}(\neg((\bigwedge_{l=0}^i \bigcirc^l \alpha) \wedge \bigcirc^{i+1} \beta), (\theta_j)_{j\in \mathbb{N}}))$, $0 \leqslant e_k <m$, by the definition of $\Phi_{k}$
			\item[] $T \vdash (\varphi \wedge \theta_k) \rightarrow \mathtt{K}_{e_k} \Phi_{k-1}(\neg((\bigwedge_{l=0}^i \bigcirc^l \alpha) \wedge \bigcirc^{i+1} \beta), (\theta_j)_{j\in \mathbb{N}})$, by propositional tautology $(p \rightarrow (q \rightarrow r)) \leftrightarrow ((p \wedge q) \rightarrow r)$.
			\item[] If we denote by $(\bar{\theta}_j)_{j\in \mathbb{N}}$ the sequence which coincides everywhere with $(\theta_j)_{j\in \mathbb{N}}$ for $j \neq k$, with the exception that $\bar{\theta}_k \equiv \varphi \wedge \theta_k$ we get that:
			\item[] $T \vdash \bar{\theta}_k \rightarrow \mathtt{K}_{e_{k-1}} \Phi_{k-1}(\neg((\bigwedge_{l=0}^i \bigcirc^l \alpha) \wedge \bigcirc^{i+1} \beta), (\bar{\theta}_j)_{j\in \mathbb{N}})$,
			\item[] $T \vdash \Phi_{k}(\neg((\bigwedge_{l=0}^i \bigcirc^l \alpha) \wedge \bigcirc^{i+1} \beta), (\bar{\theta}_j)_{j\in \mathbb{N}})$ for all $i \geqslant 0$,
			\item[] $T \vdash \Phi_{k}(\neg (\alpha \mathtt{U} \beta), (\bar{\theta}_j)_{j\in \mathbb{N}})$ by application of [RI]
			\item[] $T \vdash (\varphi \wedge \theta_k) \rightarrow \mathtt{K}_{e_{k-1}} \Phi_{k-1}(\neg (\alpha \mathtt{U} \beta), (\bar{\theta}_j)_{j\in \mathbb{N}})$
			\item[] $T \vdash \varphi \rightarrow (\theta_k \rightarrow \mathtt{K}_{e_{k-1}} \Phi_{k-1}(\neg (\alpha \mathtt{U} \beta), (\theta_j)_{j\in \mathbb{N}}))$
			\item[] $T \vdash \varphi \rightarrow \Phi_{k}(\neg (\alpha \mathtt{U} \beta), (\theta_j)_{j\in \mathbb{N}}))$
			\item[] $T \vdash \varphi \rightarrow \psi$.
		\end{itemize}
	\end{proof}
\end{theorem}

\begin{definition}
For a given set of formulae $T$ we define set $\ast T =\{\ast \alpha | \alpha \in T\}$ and $\ast \in \{\mathtt{K}_i, \bigcirc, \CIRCLE\}, 0 \leqslant i < m$. Also, for a given set of formulae $T$ we define set $\mathtt{K}_i^{-}(T)=\{\alpha|\mathtt{K}_i \alpha \in T\}$.
\end{definition}

\begin{lemma}
	Let $\alpha, \beta$ be formulae:
	\begin{itemize}
		\item[LF1] $\vdash \mathtt{G} \alpha \leftrightarrow \alpha \wedge \bigcirc \mathtt{G} \alpha$, 
		\item[LP1] $\vdash \mathtt{H} \alpha \leftrightarrow \alpha \wedge \CIRCLE \mathtt{H} \alpha$, 
		\item[LF2] $\vdash \mathtt{G} \bigcirc\alpha \leftrightarrow \bigcirc \mathtt{G} \alpha$,
		\item[LP2] $\vdash \mathtt{H} \CIRCLE\alpha \leftrightarrow \CIRCLE \mathtt{H} \alpha$,
		\item[LF3] $(\bigcirc \alpha \rightarrow \bigcirc \beta) \rightarrow \bigcirc (\alpha \rightarrow \beta)$,
		\item[LP3] $(\CIRCLE \alpha \rightarrow \CIRCLE \beta) \rightarrow \CIRCLE (\alpha \rightarrow \beta)$,
		\item[LF4] $(\bigcirc \alpha \wedge \bigcirc \beta) \leftrightarrow \bigcirc (\alpha \wedge \beta)$,
		\item[LP4] $(\CIRCLE \alpha \wedge \CIRCLE \beta) \leftrightarrow \CIRCLE (\alpha \wedge \beta)$,
		\item[LF5] $(\bigcirc \alpha \vee \bigcirc \beta) \leftrightarrow \bigcirc (\alpha \vee \beta)$,
		\item[LP5] $(\CIRCLE \alpha \vee \CIRCLE \beta) \leftrightarrow \CIRCLE (\alpha \vee \beta)$,
		\item[LF6] $\mathtt{G} \alpha \vdash \bigcirc^i \alpha$, $i \geqslant 0$, 
		\item[LP6] $\mathtt{H} \alpha \vdash (\CIRCLE \bot \wedge \alpha) \vee \CIRCLE^i \alpha$, $i \geqslant 0$,
		\item[LF7] if $\vdash \alpha$ then $\vdash \mathtt{G} \alpha$,
		\item[LP7] if $\vdash \alpha$ then $\vdash \mathtt{H} \alpha$,
		\item[LF8] if $T \vdash \alpha$, where $T$ is a set of formulae, then $\bigcirc T \vdash \bigcirc \alpha$,
		\item[LP8] if $T \vdash \alpha$, where $T$ is a set of formulae, then $\CIRCLE T \vdash \CIRCLE \alpha$,
		\item[LF9] for $j \geqslant 0$, $\bigcirc^j \beta, \bigcirc^0 \alpha, \ldots , \bigcirc^{j-1} \alpha \vdash \alpha \mathtt{U} \beta$, 
		\item[LP9] for $j \geqslant 0$, $\CIRCLE^{j} \beta, \CIRCLE^{j-1} \alpha, \ldots , \CIRCLE^{0} \alpha \vdash \alpha \mathtt{S} \beta$,
		\item[LK] if $T \vdash \gamma$, where $T$ is a set of formulae, then $\mathtt{K}_e T \vdash \mathtt{K}_e \gamma$ for any $0 \leqslant e < m$.
	\end{itemize}
	\begin{proof}
		
		\begin{itemize}
			\item[LF1] $\vdash \mathtt{G} \alpha \leftrightarrow \alpha \wedge \bigcirc \mathtt{G} \alpha$
			\begin{itemize}
				\item[] $\vdash \neg (\top \mathtt{U} \neg \alpha) \leftrightarrow \neg (\neg \alpha \vee (\top \wedge \bigcirc (\top \mathtt{U} \neg \alpha)))$ (by definition of $\mathtt{G}$ and [AT4])
				\item[] $\vdash \neg (\top \mathtt{U} \neg \alpha) \leftrightarrow  \alpha \wedge (\bot \vee \bigcirc \neg (\top \mathtt{U} \neg \alpha))$ (by [AT1])
				\item[] $\vdash \neg (\top \mathtt{U} \neg \alpha) \leftrightarrow  \alpha \wedge  \bigcirc \neg (\top \mathtt{U} \neg \alpha)$ (property of $\vee$)
				\item[] $\vdash \mathtt{G} \alpha \leftrightarrow \alpha \wedge \bigcirc \mathtt{G} \alpha$ (by definition of $\mathtt{G}$)
			\end{itemize}
			
			\item[LF2] -- LF7 The proofs are the consequences of the temporal part of the above axiomatization.
			
			\item[LF8] if $T \vdash \alpha$, where $T$ is a set of formulae, then $\bigcirc T \vdash \bigcirc \alpha$
			
			We will prove this by the induction on the length of the proof of $\alpha$ from $T$.
			
			Suppose that $\alpha$ is obtained by the inference rule [MP] from $\beta \rightarrow \alpha$ and $\beta$. Then we have:
			\begin{itemize}
				\item[] $\bigcirc T \vdash \bigcirc (\beta \rightarrow \alpha)$ (induction hypothesis)
				\item[] $\bigcirc T \vdash \bigcirc (\beta \rightarrow \alpha) \rightarrow (\bigcirc \beta \rightarrow \bigcirc \alpha)$ [AT2]
				\item[] $\bigcirc T \vdash \bigcirc \beta \rightarrow \bigcirc \alpha$ [MP]
				\item[] $\bigcirc T \vdash \bigcirc \beta$ (induction hypothesis)
				\item[] $\bigcirc T \vdash \bigcirc \alpha$ [MP]
			\end{itemize} 
			
			Similarly we can prove the case when $\alpha$ is obtained using [RTN] and [RKN].
			
			Suppose that $\alpha=\Phi_{k}(\neg (\gamma \mathtt{U} \beta), (\theta_j)_{j\in \mathbb{N}}) = \theta_k \rightarrow \mathtt{K}_i \Phi_{k-1}(\neg (\gamma \mathtt{U} \beta), (\theta_j)_{j\in \mathbb{N}})$ is obtained by the inference rule [RI]. Then:
			\begin{itemize}
				\item[] for $0 \leqslant e <m, i \geqslant 0$, $\bigcirc T \vdash \bigcirc \Phi_{k}(\neg((\bigwedge_{l=0}^i \bigcirc^l \gamma) \wedge \bigcirc^{i+1} \beta), (\theta_j)_{j\in \mathbb{N}})$, by induction hypothesis
				
				\item[] for  $0 \leqslant e <m, i \geqslant 0$, $\bigcirc T \vdash \bigcirc (\theta_k \rightarrow \mathtt{K}_e \Phi_{k-1}(\neg((\bigwedge_{l=0}^i \bigcirc^l \alpha) \wedge \bigcirc^{i+1} \beta), (\theta_j)_{j\in \mathbb{N}}))$, by definition of $\Phi_{k}$
				\item[] for  $0 \leqslant e <m, i \geqslant 0$, $\bigcirc T \vdash \bigcirc (\theta_k \rightarrow \mathtt{K}_e \Phi_{k-1}(\neg((\bigwedge_{l=0}^i \bigcirc^l \alpha) \wedge \bigcirc^{i+1} \beta), (\theta_j)_{j\in \mathbb{N}})) \rightarrow   (\bigcirc \theta_k \rightarrow \bigcirc \mathtt{K}_e \Phi_{k-1}(\neg((\bigwedge_{l=0}^i \bigcirc^l \alpha) \wedge \bigcirc^{i+1} \beta), (\theta_j)_{j\in \mathbb{N}}))$ [AT2]
				\item[] for  $0 \leqslant e <m, i \geqslant 0$, $\bigcirc T \vdash \bigcirc \theta_k \rightarrow \bigcirc \mathtt{K}_e \Phi_{k-1}(\neg((\bigwedge_{l=0}^i \bigcirc^l \alpha) \wedge \bigcirc^{i+1} \beta), (\theta_j)_{j\in \mathbb{N}})$ [MP]
				\item[] for $0 \leqslant e <m$, $\bigcirc T \vdash (\bigcirc \theta_k \rightarrow \bigcirc \mathtt{K}_e \Phi_{k-1}(\neg (\gamma \mathtt{U} \beta), (\theta_j)_{j\in \mathbb{N}}))$ [RI]
				\item[] for $0 \leqslant e <m$, $\bigcirc T \vdash (\bigcirc \theta_k \rightarrow \bigcirc \mathtt{K}_e \Phi_{k-1}(\neg (\gamma \mathtt{U} \beta), (\theta_j)_{j\in \mathbb{N}})) \rightarrow \bigcirc (\theta_k \rightarrow \mathtt{K}_e \Phi_{k-1}(\neg (\gamma \mathtt{U} \beta), (\theta_j)_{j\in \mathbb{N}}))$ (LF3)
				\item[] for $0 \leqslant e <m$, $\bigcirc T \vdash \bigcirc (\theta_k \rightarrow \mathtt{K}_e \Phi_{k-1}(\neg (\gamma \mathtt{U} \beta), (\theta_j)_{j\in \mathbb{N}}))$ [MP]
				\item[] $\bigcirc T \vdash \bigcirc \Phi_{k}(\neg (\gamma \mathtt{U} \beta), (\theta_j)_{j\in \mathbb{N}})$ , by definition of $\Phi_{k}$.
			\end{itemize}
			
			\item[LF9] for $j \geqslant 0$, $\bigcirc^j \beta, \bigcirc^0 \alpha, \ldots , \bigcirc^{j-1} \alpha \vdash \alpha \mathtt{U} \beta$
			
			By propositional reasoning we can obtain:
			\begin{align*}
			\bigcirc^j \beta, \bigcirc^0 \alpha, \ldots \bigcirc^{j-1} \alpha \vdash & \beta \vee (\alpha \wedge (\bigcirc \beta \vee (\bigcirc \alpha \wedge (\ldots (\bigcirc^{j-1}\beta \vee (\bigcirc^{j-1}\alpha \wedge \\
			& (\bigcirc^j \beta \vee (\bigcirc^j \alpha \wedge \bigcirc^{j+1}(\alpha \mathtt{U} \beta))))))\ldots))).
			\end{align*}
			
			Since 
			$$\vdash \beta \vee (\alpha \wedge (\bigcirc \beta \vee (\bigcirc \alpha \wedge (\ldots (\bigcirc^{j-1}\beta \vee (\bigcirc^{j-1}\alpha \wedge (\bigcirc^j \beta \vee (\bigcirc^j \alpha \wedge \bigcirc^{j+1}(\alpha \mathtt{U} \beta))))))\ldots))) \rightarrow \alpha \mathtt{U} \beta$$
			can be gained using [AT3], we have
			$$\bigcirc^j \beta, \bigcirc^0 \alpha, \ldots , \bigcirc^{j-1} \alpha \vdash \alpha \mathtt{U} \beta.$$
			
			\item[LP1] -- LP9 The proofs are similar to the [LF1 -- LF9] respectively.

			\item[LK] if $T \vdash \gamma$, where $T$ is a set of formulae, then $\mathtt{K}_e T \vdash \mathtt{K}_e \gamma$ for any $0 \leqslant e < m$
			
			We use the transfinite induction on the length of proof $T \vdash \gamma$. Suppose that $T \vdash \gamma$ where $\gamma \equiv \Phi_{k}(\neg (\alpha \mathtt{U} \beta), (\theta_j)_{j\in \mathbb{N}}))$ is obtained using [RI] rule. Then:
			
			\begin{itemize}
				\item[] $T \vdash \Phi_{k}(\neg((\bigwedge_{l=0}^i \bigcirc^l \alpha) \wedge \bigcirc^{i+1} \beta), (\theta_j)_{j\in \mathbb{N}})$ for all  $i \geq 0$
				\item[] $\mathtt{K}_e T \vdash \mathtt{K}_e \Phi_{k}(\neg((\bigwedge_{l=0}^i \bigcirc^l \alpha) \wedge \bigcirc^{i+1} \beta), (\theta_j)_{j\in \mathbb{N}})$ by induction hypothesis,
				\item[] $\mathtt{K}_e T \vdash \top \rightarrow \mathtt{K}_e \Phi_{k}(\neg((\bigwedge_{l=0}^i \bigcirc^l \alpha) \wedge \bigcirc^{i+1} \beta), (\theta_j)_{j\in \mathbb{N}})$ for all $i \geqslant 0$
				\item[] $\mathtt{K}_e T \vdash  \Phi_{k+1}(\neg((\bigwedge_{l=0}^i \bigcirc^l \alpha) \wedge \bigcirc^{i+1} \beta), (\bar{\theta}_j)_{j\in \mathbb{N}})$ where $(\bar{\theta}_j)_{j\in \mathbb{N}}$ is a nested $k+1$-sequence such that $\bar{\theta}_{k+1} \equiv \top$, and which coincides everywhere with  $(\bar{\theta}_j)_{j\in \mathbb{N}}$ for $j \neq k+1$
				\item[] $\mathtt{K}_e T \vdash \Phi_{k+1}(\neg (\alpha \mathtt{U} \beta), (\theta_j)_{j\in \mathbb{N}}))$ by [RI]
				\item[] $\mathtt{K}_e T \vdash \top \rightarrow \mathtt{K}_e \Phi_{k}(\neg (\alpha \mathtt{U} \beta), (\theta_j)_{j\in \mathbb{N}}))$
				\item[] $\mathtt{K}_e T \vdash \top \rightarrow \mathtt{K}_e \gamma$ 
				\item[] $\mathtt{K}_e T \vdash \mathtt{K}_e \gamma$
			\end{itemize}
		\end{itemize}
	\end{proof}
\end{lemma}

%\begin{theorem}
%	Every consistent set of formulas $T$ can be extended to a maximal consistent set $T^*$.
\maxext*	
	\begin{proof}
		Let us assume that $For = \{ \alpha_i | i \geqslant 0\}$ is the set of all formulas. The maximally consistent set $T^*$ is defined recursively, as follows:
		
		\begin{enumerate}
			\item $T_0 = T$,
			\item If $\alpha_i$ is consistent with $T_i$ then $T_{i+1} = T_i \cup \{\alpha_i\}$,
			\item If $\alpha_i$ is not consistent with $T_i$ and has form $\Phi_{k}(\neg (\alpha \mathtt{U} \beta), (\theta_j)_{j\in \mathbb{N}}))$ then
			$$T_{i+1} = T_i \cup \{\neg \alpha_i, \neg \Phi_{k}(\neg((\bigwedge_{l=0}^{n_0} \bigcirc^l \alpha) \wedge \bigcirc^{n_0+1} \beta), (\theta_j)_{j\in \mathbb{N}})\}$$
			
			where $n_0$ is a positive integer such that $T_{i+1}$ is consistent,
			\item Otherwise $T_{i+1} = T_i$,
			\item $T^* = \bigcup_{n \geqslant 0} T_n$.
		\end{enumerate}
		
		The set $T_{i+1}$ obtained by the steps 2 or 4 is obviously consistent. Let us consider the step 3.
		
		If we suppose that $\neg \Phi_{k}(\neg((\bigwedge_{l=0}^{n} \bigcirc^l \alpha) \wedge \bigcirc^{n+1} \beta), (\theta_j)_{j\in \mathbb{N}})$ is not consistent with $T_i$ for every $n \geqslant 0$ then by Deduction theorem, $T_i \vdash \Phi_{k}(\neg((\bigwedge_{l=0}^{n} \bigcirc^l \alpha) \wedge \bigcirc^{n+1} \beta), (\theta_j)_{j\in \mathbb{N}})$ for every $n \geqslant 0$, and by [RI] we have $T_i \vdash \Phi_{k}(\neg (\alpha \mathtt{U} \beta), (\theta_j)_{j\in \mathbb{N}}))$ which contradicts the assumption. Thus, the set $T_i$ obtained by the step 3 is also consistent. Also, the construction guarantees that for each $\alpha \in For$, either $\alpha \in T^*$ or  $\neg \alpha \in T^*$.
		
		To prove that that $T^*$ is deductively closed it is sufficient to prove that it is closed under the inference rules. We will only prove closeness under the inference rule [RI] since the other cases are straightforward.
		
		Suppose that $\Phi_{k}(\neg (\alpha \mathtt{U} \beta), (\theta_j)_{j\in \mathbb{N}})) \notin T^*$, while $\Phi_{k}(\neg((\bigwedge_{l=0}^{n} \bigcirc^l \alpha) \wedge \bigcirc^{n+1} \beta), (\theta_j)_{j\in \mathbb{N}}) \in T^*$ for every $n \geqslant 0$. By maximality of $T^*$, $\neg \Phi_{k}(\neg (\alpha \mathtt{U} \beta), (\theta_j)_{j\in \mathbb{N}}) \in T^*$. If $\alpha_i = \Phi_{k}(\neg (\alpha \mathtt{U} \beta), (\theta_j)_{j\in \mathbb{N}}))$, then, by the construction of $T^*$ there is $n_0$ such that $\neg \Phi_{k}(\neg((\bigwedge_{l=0}^{n_0} \bigcirc^l \alpha) \wedge \bigcirc^{n_0+1} \beta), (\theta_j)_{j\in \mathbb{N}}) \in T_i$ which contradicts the fact that $\Phi_{k}(\neg((\bigwedge_{l=0}^{n} \bigcirc^l \alpha) \wedge \bigcirc^{n+1} \beta), (\theta_j)_{j\in \mathbb{N}}) \in T^*$ for every $n \geqslant 0$.
	\end{proof}
%\end{theorem}

\begin{lemma}
	${T_j^t}$ is a maximal consistent set.
	\begin{proof}
		The proof is by induction on $t$. By hypothesis, ${T_j^0}$ is maximal and consistent. Let $t \geqslant 0$ and ${T_j^t}$ be maximal and consistent.
		
		Suppose that ${T_j^{t+1}}$ is not maximal. There is a formula $\alpha$ such that $\{\alpha, \neg \alpha\} \cap {T_j^{t+1}} = \varnothing$. Consequently, $\{\bigcirc \alpha, \bigcirc \neg \alpha\} \cap {T_j^{t}} = \varnothing$. Thus, we have that $\{\bigcirc \alpha, \neg \bigcirc  \alpha\} \cap {T_j^{t}} = \varnothing$ which is in contradiction with the maximality of ${T_j^t}$. 
		
		Suppose that ${T_j^{t+1}}$ is not consistent, i.e. ${T_j^{t+1}} \vdash \alpha \wedge \neg \alpha $, for any formula $\alpha$. By [LF8], $\bigcirc {T_j^{t+1}} \vdash \bigcirc ( \alpha \wedge \neg \alpha)$ and ${T_j^{t}} \vdash \bigcirc ( \alpha \wedge \neg \alpha)$. By [LF4] and [AT1] we can show that ${T_j^{t}} \vdash \bigcirc \alpha \wedge \neg \bigcirc \alpha$, which is in contradiction with consistency of ${T_j^t}$.
	\end{proof}
\end{lemma}

%\begin{theorem}[Strong completeness]
%	Every consistent set of formulas is satisfiable.

\strongcompletness*
	\begin{proof}
		We prove that  $\gamma \in {T_j^t}$ iff $\langle r^j, t \rangle \models \gamma$ by induction on complexity of $\gamma$.
		\begin{itemize}
			\item $\gamma \in \mathbf{N}$. This is immediate consequence of the definition of $\pi$.
			\item The proof in the cases when $\gamma$ is a negation or a conjunction is standard.
			\item $\gamma = \bigcirc \alpha$.
			
			$\langle r^j, t\rangle \models \bigcirc \alpha$ iff 
			$\langle r^j, t+1\rangle \models \alpha$ iff 
			$\alpha \in {T_j^{t+1}}$ iff
			$\bigcirc \alpha \in {T_j^t}$
			
			\item $\gamma = \alpha \mathtt{U} \beta$.
			
			Suppose that $\langle r^j, t\rangle \models \alpha \mathtt{U} \beta$. There is some $i \geqslant 0$ such that $\langle r, t +i \rangle \models \beta$ and for every $l$, $0 \leqslant l < i$, $\langle r^j, t +l\rangle \models \alpha$. By the induction hypothesis, $\beta \in {T_j^{t+i}}$, for $i \geqslant 0$, and $\alpha \in {T_j^{t+l}} $, for
			$0 \leqslant l < i$. By the construction of $\mathbb{M^*}$, we have $\bigcirc^i \beta \in{T_j^{t}}$, for $i \geqslant 0$, and  $\bigcirc^l \alpha \in {T_j^{t}}$, for $0 \leqslant l < i$. Thus, by [L9], we have that $ \alpha \mathtt{U} \beta \in {T_j^t}$.
			
			For the other direction, assume that $ \alpha \mathtt{U} \beta \in {T_j^t}$. By construction of the model $\mathbb{M^*}$, for some $i \geqslant 0$,$\bigcirc^i \beta \in{T_j^t}$, i.e. $\beta \in {T_j^{t+i}}$. Let $i_0 = \min \{i : \bigcirc^i \beta \in {T_j^t}\}$. If $i_0 = 0$, $\beta \in {T_j^t}$, and by the induction hypothesis $\langle r^j, t\rangle \models \beta$. It follows that $\langle r^j, t\rangle \models \alpha \mathtt{U} \beta$. 
			Thus, suppose that $i_0 > 0$. For every $i$ such that $0 \leqslant i < i_0$, $\bigcirc^i \beta \not\in {T_j^{t}}$, i.e. $\beta \not\in T_j^{t+i}$. So we have that $ \alpha \mathtt{U} \beta \in T_j^t$ and 
			$\beta \vee (\alpha \wedge (\bigcirc \beta \vee (\bigcirc \alpha \wedge \ldots \wedge (\bigcirc^{i_0-1} \beta \vee (\bigcirc^{i_0-1} \alpha \wedge \bigcirc^{i_0}(\alpha \mathtt{U} \beta) \ldots ) \in {T_j^t}$. It follows that for every $i < i_0$,
			$\bigcirc^i \alpha \in \langle r^j, t\rangle$, $\alpha \in \langle r^j, t+i \rangle$, and by the induction hypothesis $\langle r^j, t+i\rangle \models \alpha$. From $\langle r^j, t+i_0\rangle \models \beta$, it follows that $\langle r^j, t\rangle \models \alpha \mathtt{U} \beta$.
			
			\item $\gamma = \CIRCLE \alpha$.
			
			\begin{itemize}
				\item $t=0$
				
				If $\CIRCLE \alpha \in T_j^0$ by definition we have $\langle r^j, 0\rangle \models \CIRCLE \alpha$.
				
				Also, if $\langle r^j, 0\rangle \models \CIRCLE \alpha$ then $\CIRCLE \bot \in  T_j^0$, so:
				$\vdash \bot \rightarrow \alpha$ iff  $\vdash \CIRCLE(\bot \rightarrow \alpha)$ iff  $\vdash \CIRCLE\bot \rightarrow \CIRCLE\alpha$ iff $ \CIRCLE\bot \vdash \CIRCLE\alpha$ iff $\CIRCLE \alpha \in T_j^0$
				
				\item $t>0$
				
				$\CIRCLE \alpha \in T_j^t$ iff $\bigcirc\CIRCLE \alpha \in T_j^{t-1}$ iff $\alpha \in T_j^{t-1}$ iff $\langle r^j, t-1\rangle \models \alpha$ iff $\langle r^j, t\rangle \models \CIRCLE \alpha$
			\end{itemize}
			
			\item $\gamma = \alpha \mathtt{S} \beta$.
			\begin{itemize}
				\item $t=0$
				
				If $\alpha \mathtt{S} \beta \in  T_j^0$, by [AT6] $\beta \in T_j^0$, thus $\langle r^j, 0\rangle \models \alpha \mathtt{S} \beta$.
				
				If $\langle r^j, 0\rangle \models \alpha \mathtt{S} \beta$ then $\langle r^j, 0\rangle \models \beta$ and $\beta, \CIRCLE\bot \in T_j^0$, thus by [AT6] $\alpha \mathtt{S} \beta \in  T_j^0$.
				
				\item $t>0$
				
				Suppose that $\langle r^j, t\rangle \models \alpha \mathtt{S} \beta$. There is some $0 \leqslant i \leqslant t$ such that $\langle r^j, t-i \rangle \models \beta$ and for every $l$, $0 \leqslant l < i$, $\langle r^j, t-l \rangle \models \alpha$. Thus, 
				$\beta \in {T_j^{t-i}}$, for $0 \leqslant i \leqslant t$, and $\alpha \in {T_j^{t-l}} $, for
				$0 \leqslant l < i$. By the construction of $\mathbb{M^*}$, we have $\CIRCLE^i \beta \in{T_j^{t}}$, for $0 \leqslant i \leqslant t$, and  $\CIRCLE^l \alpha \in {T_j^{t}}$, for $0 \leqslant l < i$. Thus, by [LP9], we have that $ \alpha \mathtt{S} \beta \in {T_j^t}$.
				
				Contrariwise, $\alpha \mathtt{S} \beta \in T_j^t$. By [AT6] we have that $\alpha \wedge \CIRCLE (\alpha \mathtt{S} \beta)) \in T_j^t$. If $\beta \in  T_j^t$ then $\langle r^j, t\rangle \models \beta$ and consequently $\langle r^j, t\rangle \models \alpha \mathtt{S} \beta$. Otherwise, if  $\beta \not\in  T_j^t$ then  $\alpha \in  T_j^t$ and $\CIRCLE (\alpha \mathtt{S} \beta) \in T_j^t$. By [LF8] $\bigcirc\CIRCLE (\alpha \mathtt{S} \beta) \in T_j^{t-1}$, and by [AT8] $\alpha \mathtt{S} \beta \in T_j^{t-1}$. For some $k < t$, we have $\langle r^j, t\rangle \models \CIRCLE^k \beta$ and  $\langle r^j, t\rangle \models \CIRCLE^l \alpha$, for $0 \leqslant l < k$. Thus, we have that $ \langle r^j, t\rangle \models \alpha \mathtt{S} \beta$.
			\end{itemize}
			
			\item $\gamma = \mathtt{K}_i \alpha$. 
			
			Suppose $\mathtt{K}_i \alpha \in {T_j^t}$, then $\alpha \in \mathtt{K}_i^{-}({T_{j}^{t}})$. Also, ${T_j^t} \supset \mathtt{K}_i(\mathtt{K}_i^{-}({T_{j}^{t}}))$, so for each $\langle r^{j'}, t' \rangle$ such that $\langle r^j, t \rangle \mathcal{K}_i  \langle r^{j'}, t' \rangle$ (by the definition of relation $\mathcal{K}_i$), $\langle r^{j'}, t'\rangle \models \alpha$. By induction hypothesis ($\alpha$ is subformula of $\mathtt{K}_i \alpha$), we have that $\langle r^j, t\rangle \models \mathtt{K}_i \alpha$.
			
			Conversely, let $\langle r^j, t\rangle \models \mathtt{K}_i \alpha$ and assume the opposite, i.e. that $\mathtt{K}_i \alpha \notin {T_j^t}$. Then $\mathtt{K}_i^{-}({T_j^t}) \cup \{\neg \alpha\}$ is consistent. Otherwise, by Deduction theorem $\mathtt{K}_i^{-}({T_j^t}) \vdash \alpha$, by [LK], and ${T_j^t} \supset \mathtt{K}_i(\mathtt{K}_i^{-}({T_{j}^{t}})) \vdash \mathtt{K}_i \alpha$, by maximality of $T_j^t$, and $\mathtt{K}_i \alpha \in T_j^t$ which is a contradiction. Thus, $\mathtt{K}_i^{-}({T_j^t}) \cup \{\neg \alpha\}$ can be extended to a maximal consistent set $T_{j'}^{t'}$, and:
			$$n_i \in T_j^t \Rightarrow \mathtt{K}_i n_i \in T_j^t \Rightarrow n_i \in \mathtt{K}_i({T_j^t}) \Rightarrow n_i \in T_{j'}^{t'}.$$ 
			
			Similarly, for $\neg n_i \in T_j^t$. Thus, we have $\langle r^j, t \rangle \mathcal{K}_i  \langle r^{j'}, t' \rangle$. Since $\neg \alpha \in T_{j'}^{t'}$, then $\langle r^{j'}, t' \rangle \models \neg \alpha$ by induction hypothesis, so $\langle r^{j'}, t' \rangle \not\models \alpha$, which is a contradiction.
			
		\end{itemize}
	\end{proof}
%\end{theorem}

%\begin{lemma}
%		Let a peer join a Chord network, between two nodes which constitute a stable pair, and there are not other nodes that constitute that stable pair. Then, there is a
%	finite period of time after the starting pair will be stable again.

\join*
	\begin{proof}
		Let us assume that $n_i, n_j \in \mathbf{N_a}$ and $n_i \Cap n_j$, i.e. $(n_i \succ n_j) \wedge (n_j \prec n_i)$ and that $n_k$ tries to join that stable pair. Let us denote $$\alpha = \CIRCLE(n_i \Cap n_j) \wedge \rho_{J,k}\bigwedge_{n_l \in I} \bigwedge_{t=0}^{5f} \bigcirc^t \neg n_l, I=\{n_l | n_l\mathtt{M}\langle n_i, n_j\rangle, n_k \neq n_l, n_j \neq n_l\}$$.
		
		We have,
		\begin{align}
		\alpha&\vdash \mathtt{K}_i (n_i \succ n_j) \wedge \mathtt{K}_j (n_j \prec n_i) \wedge n_k\mathtt{M}\langle n_i, n_{j}\rangle \wedge n_k \tag*{\text{(by AS6) }(1)}\\
		\alpha&\vdash \mathtt{K}_i ( n_i \succ n_j) \tag*{\text{(by definition of $\wedge$  and 1) } (2a)}\\
		\alpha&\vdash \mathtt{K}_j (n_j \prec n_i ) \tag*{\text{(by definition of $\wedge$  and 1) } (2b)}\\
		\alpha&\vdash n_k\mathtt{M}\langle n_i, n_{j}\rangle \tag*{\text{(by definition of $\wedge$  and 1) } (2c)}\\
		\alpha&\vdash n_k \tag*{\text{(by definition of $\wedge$  and 1) } (2d)}\\
		\alpha&\vdash \rho_{J,k} \tag*{\text{(by definition of $\alpha$) } (2e)}\\
		\alpha&\vdash n_k \wedge \rho_{J,k} \tag*{\text{(by 2d, 2e) } (2f)}\\
		\alpha&\vdash \rho_{J,k} \rightarrow \bigvee_{l=0}^f \bigcirc^l \mathtt{K}_k (n_k \succ n_j) \tag*{\text{(by definition  of $\rho_{J,k}$)} (3)} \\
		\alpha&\vdash \bigvee_{l=0}^f \bigcirc^l \mathtt{K}_k (n_k \succ n_j) \tag*{\text{(by MP, 2e, 3)} (4)}\\
		\alpha&\vdash \bigvee_{l=0}^f \bigcirc^l \mathtt{K}_k (n_k \succ n_j)\rightarrow \Circle^f \mathtt{K}_k (n_k \succ n_j), \tag*{\text{(by definition  of AS6,4)} (5)} \\
		\alpha&\vdash \Circle^f \mathtt{K}_k (n_k \succ n_j) \tag*{\text{(by MP, 4, 5)} (6)} \\
		\alpha&\vdash n_k \wedge \rho_{J,k} \rightarrow \bigvee_{l=0}^f \bigcirc^l \rho_{S1,k,j} \tag*{\text{[ACF3]} (7)}\\
		\alpha&\vdash \bigvee_{l=0}^f \bigcirc^l \rho_{S1,k,j} \tag*{(by MP, 2f,7) (8a)}\\
		\alpha&\vdash \bigcirc^f \rho_{S1,k,j} \tag*{(by AS6) (8b)}\\
		\alpha&\vdash \bigcirc^f ((\mathtt{K}_k (n_k \succ n_{j}) \wedge \mathtt{K}_j (n_j \prec n_i) \wedge n_k\mathtt{M}\langle n_i, n_{j}\rangle)
		\rightarrow \bigvee_{l=0}^f \bigcirc^l \mathtt{K}_j (n_j \prec n_{k})) \tag*{\text{(by 8b)} (9a)}\\
		\alpha&\vdash \Circle^{f} ((\mathtt{K}_k (n_k \succ n_{j}) \wedge \mathtt{K}_j (n_j \prec n_i) \wedge n_k\mathtt{M}\langle n_i, n_{j}\rangle))
				\rightarrow \Circle^{f}(\bigvee_{l=0}^f \bigcirc^l \mathtt{K}_j (n_j \prec n_{k})) \tag*{\text{(by AT2, 9a)} (9b)}\\
		\alpha&\vdash \Circle^{f} \mathtt{K}_k (n_k \succ n_j) \tag*{(by AS6, 6) (10a)}\\
		\alpha&\vdash \Circle^{f} \mathtt{K}_k (n_j \prec n_i) \tag*{(by AS6, 2b) (10b)}\\
		\alpha&\vdash \Circle^{f} (n_k\mathtt{M}\langle n_i, n_{j}\rangle) \tag*{(by AS6, 2c) (10c)}\\
		\alpha&\vdash \Circle^{f} (\mathtt{K}_k (n_k \succ n_{j}) \wedge \mathtt{K}_j (n_j \prec n_i) \wedge n_k\mathtt{M}\langle n_i, n_{j}\rangle)
		\tag*{\text{(by LF4, 10a, 10b, 10c)} (11)}\\
		\alpha&\vdash \Circle^{f}(\bigvee_{l=0}^f \bigcirc^l \mathtt{K}_j (n_j \prec n_{k})) \tag*{(by MP, 9,11) (12)}\\
		\alpha&\vdash \Circle^{2f} \mathtt{K}_j (n_j \prec n_{k}) \tag*{\text{(by definition  of $\Circle$, AS6, 12)} (13)}\\
		\alpha&\vdash \bigvee_{l=f}^{2f} \bigcirc^l\bigvee_{j=0}^{m-1} \rho_{S2,i,j} \tag*{\text{(by $n_i \in \mathbf{N_a}$ and ACF2 or ACF4)} (14)}\\
		\alpha&\vdash \bigvee_{l=f}^{2f}  \bigcirc^l \rho_{S2,i,k} \tag*{\text{(by definition  of $\vee$,14)} (15a)}\\
		\alpha&\vdash \bigcirc^{2f} \rho_{S2,i,k} \tag*{\text{(by definition  AS6,15a)} (15b)}\\
		\alpha&\vdash \Circle^{2f} (\mathtt{K}_i (n_i \succ n_{j}) \wedge \mathtt{K}_j (n_j \prec n_{k}) \wedge n_k\mathtt{M}\langle n_i, n_{j}\rangle) \rightarrow \bigvee_{l=0}^f \bigcirc^l\mathtt{K}_i (n_i \succ n_k)) \tag*{\text{(by 15b)} (16a)}\\
		\alpha&\vdash \Circle^{2f} (\mathtt{K}_i(n_i \succ n_{j}) \wedge \mathtt{K}_j (n_j \prec n_{k}) \wedge n_k\mathtt{M}\langle n_i, n_{j}\rangle)) \rightarrow \Circle^{2f}(\bigvee_{l=0}^f \bigcirc^l \mathtt{K}_i (n_i \succ n_k)) \tag*{\text{(by AT2, 16a)} (16b)}\\
		\alpha&\vdash \Circle^{2f} \mathtt{K}_i (n_i \succ n_{j}) \tag*{(by AS6, 2a) (17a)}\\
		\alpha&\vdash \Circle^{2f} \mathtt{K}_j (n_j \prec n_{k}) \tag*{(by AS6, 13) (17b)}\\
		\alpha&\vdash \Circle^{2f} (n_k\mathtt{M}\langle n_i, n_{j}\rangle) \tag*{(by AS6, 2c) (17c)}\\
		\alpha&\vdash \Circle^{2f} (\mathtt{K}_i (n_i \succ n_{j}) \wedge \mathtt{K}_j (n_j \prec n_{k}) \wedge n_k\mathtt{M}\langle n_i, n_{j}\rangle) \tag*{\text{(by LF4, 17a, 17b, 17c)} (18)}\\
		\alpha&\vdash \Circle^{2f}(\bigvee_{l=0}^f \bigcirc^l\mathtt{K}_i (n_i \succ n_k)) \tag*{(by MP, 16b,18) (19)} \\
		\alpha&\vdash \Circle^{3f}\mathtt{K}_i (n_i \succ n_k),  \tag*{\text{(by definition  of $\Circle $, AS6,19)} (20)}\\
		\alpha&\vdash \bigvee_{l=3f}^{4f}  \bigcirc^l \bigvee_{j=0}^{m-1} \rho_{S1,i,j} \tag*{\text{(by $n_i \in \mathbf{N_a}$ and ACF1 or ACF3)} (21)}\\
		\alpha&\vdash \bigvee_{l=3f}^{4f} \bigcirc^l \rho_{S1,i,k} \tag*{\text{(by definition  of $\vee$,21)} (22a)} \\
		\alpha&\vdash  \bigcirc^{4f}\rho_{S1,i,k} \tag*{\text{(by definition  AS6, 22a)} (22b)} \\
		\alpha&\vdash \Circle^{4f} (\mathtt{K}_i(n_i \succ n_{k}) \wedge \mathtt{K}_k(n_k \prec u) \wedge n_k\mathtt{M}\langle n_i, n_{j}\rangle) \rightarrow \bigvee_{l=0}^f \bigcirc^l\mathtt{K}_k(n_k \prec n_{i}))  \tag*{\text{(by  22b)} (23a)}\\
		\alpha&\vdash \Circle^{4f} (\mathtt{K}_i (n_i \succ n_{k}) \wedge \mathtt{K}_k (n_k \prec u) \wedge n_k\mathtt{M}\langle n_i, n_{j}\rangle)) \rightarrow \Circle^{4f}(\bigvee_{l=0}^f \bigcirc^l\mathtt{K}_k(n_k \prec n_{i})) \tag*{\text{(by definition  of AT2, 23a)} (23b)}\\
		\alpha&\vdash \Circle^{4f} \mathtt{K}_i(n_i \succ n_{k}) \tag*{(by AS6, 29) (24a)}\\
		\alpha&\vdash \Circle^{4f} \mathtt{K}_k (n_k \prec u) \tag*{(by AS6, 2e) (24b)}\\
		\alpha&\vdash \Circle^{4f} (n_k\mathtt{M}\langle n_i, n_{j}\rangle) \tag*{(by AS6, 2c) (24c)}\\
		\alpha&\vdash \Circle^{4f} (\mathtt{K}_i (n_i \succ n_{k}) \wedge \mathtt{K}_k(n_k \prec u) \wedge n_k\mathtt{M}\langle n_i, n_{j}\rangle) \tag*{\text{(by LF4, 24a, 24b, 24c)} (25)}\\
		\alpha&\vdash \Circle^{4f}(\bigvee_{l=0}^f \bigcirc^l \mathtt{K}_k(n_k \prec n_{i})) \tag*{(by MP, 23b,25) (26)} \\
		\alpha&\vdash \Circle^{5f}\mathtt{K}_k(n_k \prec n_{i}) \tag*{\text{(by definition  of $\Circle$, AS6,26)} (27)}\\
		\alpha&\vdash \Circle^{5f} \mathtt{K}_k(n_k \succ n_j)  \tag*{\text{(by AS6, 6)} (28)}\\
		\alpha&\vdash \Circle^{5f} \mathtt{K}_j (n_j \prec n_{k}) \tag*{\text{(by AS6, 13)} (29)}\\
		\alpha&\vdash \Circle^{5f} \mathtt{K}_i(n_i \succ n_k) \tag*{\text{(by AS6, 20)} (30)} \\
		\alpha&\vdash \Circle^{5f} \mathtt{K}_k(n_k \prec n_{i}) \tag*{\text{(by AS6, 27)} (31)} \\
		\alpha&\vdash \Circle^{5f}(\mathtt{K}_k(n_k \succ n_j) \wedge \mathtt{K}_j (n_j \prec n_{k}) \wedge \mathtt{K}_i (n_i \succ n_k) \wedge \mathtt{K}_k (n_k \prec n_{i}))\tag*{\text{(by LF4, 28, 29, 30, 31)} (32)} \\
		\alpha&\vdash \Circle^{5f} (n_i \Cap n_j) \tag*{\text{(by definition of $\Cap$)} (33)}
		\end{align}
	\end{proof}
%\end{lemma}

\section*{Acknowledgment}

The work presented here was supported by Serbian Ministry of
Education, Science and Technology Development (the projects ON174026 and III44006),
through Matematički institut SANU, and Ministarstvo znanosti,
obrazovanja i športa republike Hrvatske. This work was supported in part by the Slovenian Research Agency within the research program Algorithms and Optimization Methods in Telecommunications.

\end{document}